\documentclass[11pt]{amsart}
\usepackage[margin=1in]{geometry}                % See geometry.pdf to learn the layout options. There are lots.
\usepackage{graphicx}
\usepackage{amssymb}
\usepackage{enumitem}
\usepackage{epstopdf}
\usepackage[dvipsnames]{xcolor}
\usepackage{tikz}
\usepackage{algpseudocode}
\usepackage{mathtools}
\usepackage{hyperref}
\usepackage{mathrsfs}
\usepackage{caption}

\usepackage{color}

\newcommand{\todo}[1]{{}}
\newcommand{\na}[1]{{}}
\newcommand{\remove}[1]{} %%%Change this to #1 to put these back

\newtheorem{theorem}{Theorem}
\newtheorem{lemma}{Lemma}

\newtheorem{definition}{Definition}

\newtheorem{proposition}{Proposition}
\newtheorem{corollary}{Corollary}

\definecolor{carmine}{rgb}{0.59, 0.0, 0.09}
\newcommand{\sref}[1]{\hyperref[#1]{\color{blue} \ref{#1}}}
\newcommand{\tref}[1]{\hyperref[#1]{\color{Mahogany} \ref{#1}}}
\newcommand{\demand}{D}

\newcommand{\fref}[1]{\hyperref[#1]{\color{PineGreen} \ref{#1}}}

\DeclareMathOperator*{\E}{\mathbb{E}}
\newcommand{\p}{\mathbb{P}}

\renewcommand{\P}{\mathcal{P}}
\newcommand{\mI}{\mathcal{I}}
\newcommand{\N}{\mathbb{N}}

\newcommand{\mA}{\mathcal{A}}
\newcommand{\Interest}{\mathfrak{I}}
\newcommand{\Admissions}{\mathfrak{A}}
\newcommand{\Matching}{\mathfrak{M}}

\newcommand{\M}{\mathcal{M}}
\newcommand{\R}{\mathcal{R}}
\newcommand{\mE}{\mathcal{E}}
\newcommand{\indif}{\Theta}

\newcommand{\averagerank}{AverageRank}
\newcommand{\acceptancerate}{AcceptanceRate}

\newcommand{\mv}{\mathcal{V}}

\renewcommand{\S}{\ensuremath{\mathcal{S}}}
\renewcommand{\H}{\ensuremath{\mathcal{H}}}

\newcommand{\pois}[1]{\ensuremath{{\bf Po}(#1)}}

\newcommand{\abs}[1]{\ensuremath{\left| #1 \right|}}
\newcommand{\norm}[1]{\ensuremath{\left|\left| #1 \right|\right|_1}}

\newcommand{\enrollment}{Enrollment}

\newcommand{\T}{\mathcal{T}}
\newcommand{\D}{\mathcal{D}}

\newcommand{\mV}{\mathcal{V}}
\newcommand{\vdet}{\mV^{det}}
\newcommand{\vpois}{\mV^{pois}}

\newcommand{\Dtau}{\D^\tau}

\newcommand{\I}{I}

\newcommand{\MStable}{M^*}

\newcommand{\interest}{I}

\usepackage{multirow}
\usepackage{comment}
\usepackage{natbib}
\usepackage{url}
\usepackage{setspace}
\usepackage{bm}
\title{A Continuum Model of \\ Stable Matching with Finite Capacities}
\author{Nick Arnosti, University of Minnesota }

\begin{document}

%Determining location of subscript
$\fontdimen14\textfont2=4pt
\fontdimen15\textfont2=4pt
\fontdimen16\textfont2=3pt
\fontdimen17\textfont2=3pt$

\onehalfspacing

\maketitle
\begin{abstract}
This paper introduces a unified framework for stable matching, which nests the traditional definition of stable matching in finite markets and the continuum definition of stable matching from \citet{azevedo-leshno_2016} as special cases. Within this framework, I identify a novel continuum model, which makes individual-level probabilistic predictions.

This new model always has a unique stable outcome, which can be found using an analog of the Deferred Acceptance algorithm. The crucial difference between this model and that of \citet{azevedo-leshno_2016} is that they assume that the amount of student interest at each school is deterministic, whereas my proposed alternative assumes that it follows a Poisson distribution. As a result, this new model accurately predicts the simulated distribution of cutoffs, even for markets with only ten schools and twenty students. 

This model generates new insights about the number and quality of matches. When schools are homogeneous, it provides upper and lower bounds on students' average rank, which match results from \citet{ashlagi-kanoria-leshno_2017} but apply to more general settings. This model also provides clean analytical expressions for the number of matches in a platform pricing setting considered by \citet{marx-schummer_2021}. 
%We introduce a unified framework for stable matching, which nests the traditional definition of stable matching in finite markets and the continuum definition of stable matching from \citet{azevedo-leshno_2016} as special cases. Within this framework, we introduce a novel continuum model, which makes individual-level probabilistic predictions.
%
%Our model always has a unique stable outcome, which can be found using an analog of the Deferred Acceptance algorithm. The crucial difference between our model and that of \citet{azevedo-leshno_2016} is that they assume that the amount of student interest at each school is deterministic, whereas we assume that it follows a Poisson distribution. As a result, our model accurately predicts the simulated distribution of cutoffs, even for markets with only ten schools and twenty students. 
%
%We use our model to generate new insights about the number and quality of matches. When schools are homogeneous, we provide upper and lower bounds on students' average rank, which match results from \citet{ashlagi-kanoria-leshno_2017} but apply to more general settings. Our model also provides clean analytical expressions for the number of matches in a platform pricing setting considered by \citet{marx-schummer_2021}. 
\end{abstract}

\na{Alex Teytelboym's point: don't want to be seen as incremental improvements on Azevedo Leshno. Point out their serious shortcomings. Point out serious shortcomings of Ashlagi Kanoria Leshno (lack of flexibility). Say that I am (sort of) unifying these models in the right way.
}

\section{Introduction}

Ever since \citet{gale-shapley_1962} defined stability in two-sided matching markets, the topic has generated a great deal of interest from academics and practitioners alike: their paper has over 7,000 citations, and variants of their deferred acceptance algorithm are used to assign medical residencies in the United States and public school seats in cities across the globe. These developments prompted the award of the 2012 Nobel Prize to Alvin Roth and Lloyd Shapley ``for the theory of stable allocations and the practice of market design."

Although many aspects of this theory are now well-understood, it remains difficult to predict how changes to market primitives will affect the set of stable outcomes. A recent report by the Brookings institution highlighted this as a key challenge facing school choice initiatives:
%A recent report by the Brookings institution highlights this challenge as a key barrier to successful implementation of school choice initiatives:
\begin{quote}
{\it Even if DA [Deferred Acceptance] algorithms are relatively simple, predicting how student assignment policies will affect enrollment and outcomes is difficult... 
%What happens if the priority applies only to a subset of the seats in a school? What are the consequences of putting it higher or lower in the hierarchy of priorities? These questions are hard to answer without considerable analysis... All of 
This creates challenges for policymakers to assess a priori how policy decisions will affect students and schools -- and creates potential for unintended negative consequences. \hfill \citep{kasman-valant_2019} }
\end{quote}

Along these lines, \citet{kojima_2012} shows that affirmative action policies that boost the priority of students in a targeted group may result in worse outcomes for all members of this group. Even when the direction of an effect is known, seemingly small changes can have surprisingly large consequences: \citet{ashlagi-kanoria-leshno_2017} show that in balanced markets with random preferences, adding a single student dramatically increases students' average rank. Conversely, seemingly large changes may have minimal impact. In Boston, a 1999 compromise granted students  living within a school's walk zone priority for 50\% of its seats. Over a decade later, \citet{dur-kominers-pathak-sonmez_2018} demonstrated that this policy led to virtually identical outcomes as a policy that granted no walk zone priority at all!

%\subsection{\na{Prior Approaches}}

In an effort to better understand two-sided matching markets, researchers have turned to a variety of ``large market" and ``continuum" models. An ideal approach would be {\em flexible} enough to incorporate complex preferences and priorities, able to {\em accurately} predict match outcomes, and would offer new {\em insights}. 

One line of work studies outcomes in large  finite markets. To maintain tractability, these papers typically impose strong assumptions. For example,\citet{pittel_1989}, \citet{knuth_1996}, \citet{ashlagi-kanoria-leshno_2017}, \citet{ashlagi-nikzad-romm_2019}, \citet{ashlagi-nikzad_2019} and \citet{kanoria-min-qian_2021} all assume that schools are symmetric (student preferences are iid and uniformly distributed), and all but \citet{knuth_1996} assume that school priorities are either identical or drawn independently and uniformly at random. Recent work by \citet{arnosti_2021} permits schools to differ along a one-dimensional measure of ``popularity," but still imposes strong assumptions on preferences and priorities. These papers offer {\em insights} into how match outcomes depend on the choice of proposing side, the market imbalance, and the length of student lists. However, the analysis underlying these insights is not {\em flexible} enough to accommodate more realistic assumptions on preferences and priorities. 

As a result, these models are of limited use when trying to tackle practical problems. For example, parents may want to know how likely their child is to be admitted to a particular school. Administrators may want to predict how a proposed policy change will affect the number of students who fail to match to any school on their list. For these problems, the best tool is often simulation. \citet{abdulkadiroglu-pathak-roth_2009} use simulations to compare different tiebreaking procedures in New York City and Boston, and \citet{dehaan-gautier-oosterbeek-vanderklaauw_2018} do the same for Amsterdam. \citet{ashlagi-nikzad_2019} and \citet{kanoria-min-qian_2021} use simulation data from New York to answer different questions. While simulation is a very {\em flexible} and {\em accurate} tool, it typically does not offer much {\em insight}. It shows what is true, but not why it is true. If a pattern is observed through simulation, it can be difficult to predict whether the same pattern will continue to hold in other settings.

\subsection{The Continuum Model of \citet{azevedo-leshno_2016}} 

Perhaps the work that comes closest to hitting the trifecta of flexibility, accuracy, and insightfulness is that of \citet{azevedo-leshno_2016}. Students in their model are described by a ``type" $\theta$, which determines both their preferences and their priority score at each school. Student types are distributed according to an (almost) arbitrary measure $\eta$, giving the model {\em flexibility} to capture complex preferences and priorities. This measure determines the set of market-clearing ``cutoff scores":  school-specific scores such that if schools admit students with priority above their cutoff score, and students attend their favorite school where they are admitted, expected ``demand" (enrollment) is equal to capacity at each school with a positive cutoff. This approach enables {\em tractable} analysis into the effect of market primitives on the final cutoffs.

%describe stable matchings with a cutoff score for each school. Students are admitted to each school where their priority exceeds the cutoff, and match to their favorite such school.

The {\em accuracy} of their model at predicting outcomes in finite random markets depends on the number of seats at each school. We elaborate on this point, as it motivates our work. Hypothetically, if each school were to use its predicted cutoff score to determine admissions, and $n$ student types were drawn iid from the measure $\eta$, then a school $h$ with expected demand equal to its capacity $C_h$ would have {\em realized} demand following a binomial distribution with parameters $n$ and $C_h/n$. In reality, capacity constraints must hold for each realization (rather than only in expectation), so the fluctuations in realized demand translate to fluctuations in realized cutoff scores. As a result, a student whose priority is a bit above the predicted cutoff score is not truly ``safe", and one with priority below the predicted cutoff is not hopeless. Rather, students' probability of admission varies continuously with their priority.

If $C_h$ is large, then a Binomial with parameters $n$ and $C_h/n$ will be concentrated around its expectation, and the fluctuation in cutoff scores will be minor. That is, students with priority significantly above the predicted cutoff are almost certain to be admitted, while those with priority significantly below the predicted cutoff have almost no chance. This is formalized in Theorem 2 and Proposition 3 from \citet{azevedo-leshno_2016}.

However, the capacity of each school is small, then random variation in student demand results in significant variation in schools' realized cutoff scores, as shown in Figure \tref{fig:azevedo-leshno}. Their model fails to capture this variability, causing it to produce inaccurate predictions. For a simple (if stylized) illustration of this point, consider a market where $n$ schools each have capacity $C$ and $nC$ students each list a single school selected uniformly at random. Their model predicts a cutoff score of zero for each school, implying that all students will match to their first choice. If $C$ is very large, this is nearly correct, but if $C = 1$, the expected fraction of students who match is $1 - (1 - 1/n)^n \approx 1 - 1/e$.

% present a fluid limit model for stable matching. Given a distribution over student preferences and priorities, they seek ``cutoff scores" at each school that ``clear the market:" the expected number of students matched to each school should equal the number of seats available at that school. This model makes it possible to study markets with complex correlations between student preferences and priorities. Its usefulness is demonstrated by Theorem 1 in their paper, which establishes that there is generically a unique stable matching.

\begin{figure}
\captionsetup{width=\textwidth}
%\centerline{\includegraphics[width = .33\textwidth]{A-L-Fig4a-dup}
%\includegraphics[width = .33\textwidth]{A-L-Fig4b-dup}
%\includegraphics[width = .33\textwidth]{A-L-Fig4c-dup}}
\includegraphics[width = \textwidth]{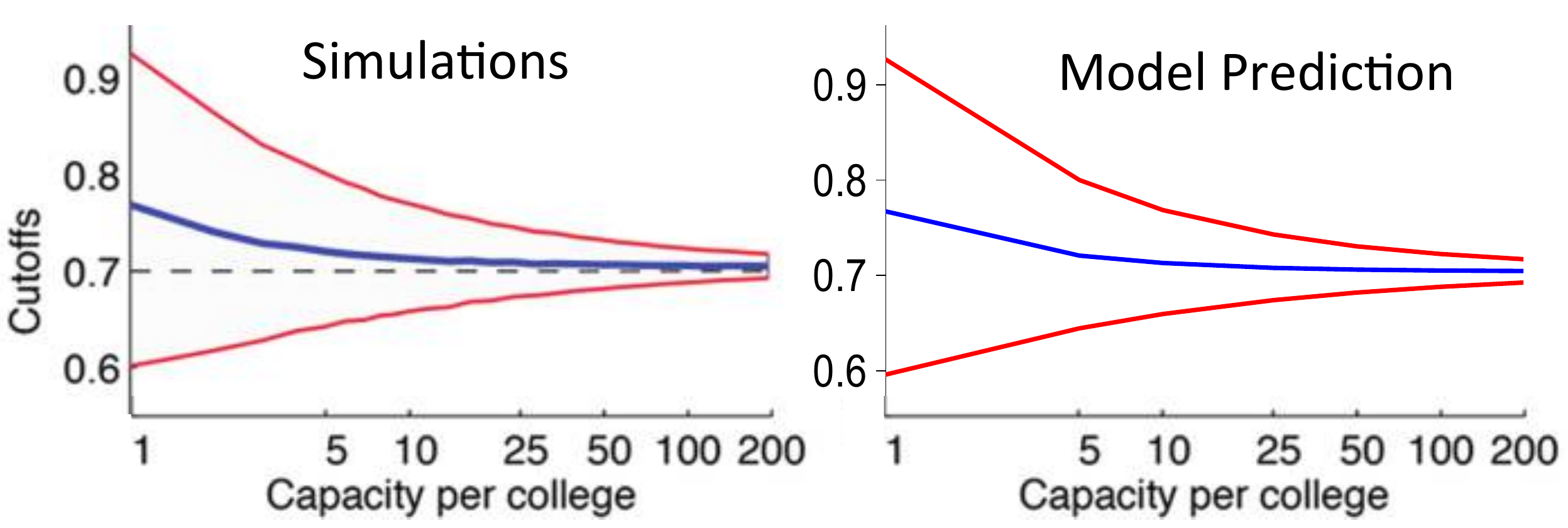}
\caption{The distribution of ``cutoff scores" required for admission in a class of examples considered by \citet{azevedo-leshno_2016}. In this market there are ten schools, and the $x$-axis denotes the number of seats per school, which ranges from $1$ to 200. In all scenarios, there are twice as many students as seats. Students submit complete lists drawn uniformly at random, and schools' priority scores are imperfectly correlated: they consist of the average of the student's quality (drawn uniformly on $[0,1]$) and iid student-school terms (also drawn uniformly on $[0,1]$). This correlation renders direct analysis of the finite random market intractable. The left panel displays simulation results reported by \citet{azevedo-leshno_2016}: the blue line shows the empirical average cutoff score, with the red lines representing the $5^{th}$ and $95^{th}$ percentile of the empirical distribution. Their continuum model predicts a deterministic cutoff score, shown by the black dotted line. Note that this prediction does not depend on the number of seats at each school, and does not capture the uncertainty in cutoffs, which is significant unless school capacities are large. Our proposed alternative predicts the {\em distribution} of cutoff scores. The right panel shows the average, $5^{th}$ percentile and $95^{th}$ percentile of the predicted distribution. \todo{Is it better to put in full panel?}
%\todo{ \citet{azevedo-leshno_2016} present simulation results from a family of markets with twice as many students as school seats. In their simulations, the number of schools ranges from $10$ to $500$, and the number of seats at each school from $1$ to $200$. . 
%}
%Bottom: simulation results from \citet{azevedo-leshno_2016}. The blue line shows the average cutoff, with the red lines representing the $5^{th}$ and $95^{th}$ percentile of the empirical cutoff distribution. The model of \citet{azevedo-leshno_2016} makes a point prediction, shown by the black dotted line. This is accurate when capacities are large, but is biased and does not capture uncertainty. Our model predicts a cutoff distribution: the top panel shows the predicted average, $5^{th}$ percentile and $95^{th}$ percentile of cutoffs.
} \label{fig:azevedo-leshno}
\end{figure}

\subsection{Our Contributions}

In Section \ref{sec:poisson}, we introduce a model that addresses this shortcoming, and predicts a distribution (rather than a point estimate) for each school's cutoff score. As noted above, if $D_h(P)$ denotes {\em expected} demand at school $h$ when all schools use cutoff scores given by $P$, then {\em realized} demand at $h$ will follow a binomial distribution with parameters $n$ and $D_h(P)/n$. If the number of students $n$ is modestly large, this is well-approximated by a Poisson distribution with mean $D_h(P)$. Motivated by these observations, our model assumes that for any potential cutoff $p \in [0,1]$, the realized demand at school $h$ follows a Poisson distribution. 

%From this assumption, we wish to predict the cutoff score distribution at each school. To do so, we need to know the expected demand at all potential cutoffs (not just a single vector of cutoff scores). We describe this using an interest function $I$

%We use ``interest functions" to describe expected demand at different potential cutoffs. 

We use this assumption to calculate a cutoff score distribution for each school. Of course, the calculation for school $h$ does not assume that schools other than $h$ use deterministic cutoff scores. Instead, we look for a set of ``self-consistent" cutoff score distributions (in a sense made precise in Section \sref{sec:stability}).

Our model, like that of \citet{azevedo-leshno_2016}, is {\em flexible} enough to permit a nearly arbitrary joint distribution of student preferences and priorities. We demonstrate its {\em tractability} by reproducing key insights from the recent work of \citet{ashlagi-kanoria-leshno_2017} in Section \sref{sec:avg-rank} (and extending these insights to settings where students submit incomplete lists and schools have multiple seats), and by providing closed-form expressions for the number of matches in a setting considered by \citet{marx-schummer_2021} in Section \sref{sec:matched-students}. Finally, our new model is {\em accurate}. As school capacities grow, its predictions converge to those of \citet{azevedo-leshno_2016}. When capacities are not large, we demonstrate numerically that our model accurately predicts a range of outcomes: Figure \tref{fig:azevedo-leshno} compares its predictions for the distribution of school cutoffs to simulations from \cite{azevedo-leshno_2016}, Figure \tref{fig:akl2} compares its predictions for students' average rank to simulations from \citet{ashlagi-kanoria-leshno_2017}, and Figure \fref{fig:ms12} compares its predictions for the number of matches formed to exact results and simulations from \citet{marx-schummer_2021}. In all three cases, the match is excellent. % (even in a market with only twenty students and ten schools).

% and \citet{azevedo-leshno_2016}, respectively. Figure \tref{fig:akl2} demonstrates that our model not only makes the qualitative prediction that average rank increases significantly as soon as students outnumber schools, but also makes extremely accurate {\em quantitative} predictions about the magnitude of this effect. Meanwhile, Figure \tref{fig:azevedo-leshno} demonstrates that not only does our model predict a distribution of school cutoffs (rather than a single value), 
 
 %and can be used to predict how likely a particular student is to be admitted to a particular school
In addition to providing a new model which is flexible, tractable, and accurate, we provide a new framework for studying stable matching. This framework describes matchings through three different perspectives: {\em admissions functions} which describe the distribution of school cutoff scores, {\em interest functions} which describe the expected number of students at each priority level who wish to attend each school, and {\em matchings} which describe the probability that students are assigned to each school on their list.

Our framework identifies a large class of stable matching models, parameterized by a {\em measure of student types $\eta$} and a {\em vacancy function} $\mV$. We show that suitable choices of these parameters recover the traditional model of stable matching in finite markets (Proposition \tref{prop:finite}) as well as the continuum model of \citet{azevedo-leshno_2016} (Proposition \tref{prop:al}). Both of these definitions use a deterministic vacancy function $\vdet$ with range $\{0,1\}$. By contrast, our model from Section \sref{sec:poisson} uses a vacancy function $\vpois$ based on the Poisson distribution, which allows it to make probabilistic predictions.

In addition to offering a unified perspective on several definitions of stable matching, our framework identifies other definitions associated with other choices of the vacancy function $\mV$. We show that for {\em any} weakly decreasing vacancy function, classic results for stable matchings continue to hold: stable matchings always exist and form a lattice (Theorem \tref{thm:existence}), and the extreme points of this lattice can be found using a generalization of the Deferred Acceptance algorithm (defined in Section \sref{sec:da-general}). If $\eta$ has no mass points, then the ``rural hospital theorem" holds (Theorem \tref{thm:rht}). Finally, if the vacancy function $\mV$ is {\em strictly} decreasing, then there is a unique stable matching (Theorem \tref{thm:uniqueness}). This uniqueness result holds more generally than that of \citet{azevedo-leshno_2016}, and helps to explain the small core observed empirically by \citet{roth-peranson_1999}, and theoretically by \citet{immorlica-mahdian_2005}, \citet{kojima-pathak_2009}, and \citet{ashlagi-kanoria-leshno_2017}.

\begin{figure}
\captionsetup{width=\textwidth}
\centerline{%\includegraphics[height = 2.4 in]{A-K-L-Fig2}
\includegraphics[width = \textwidth]{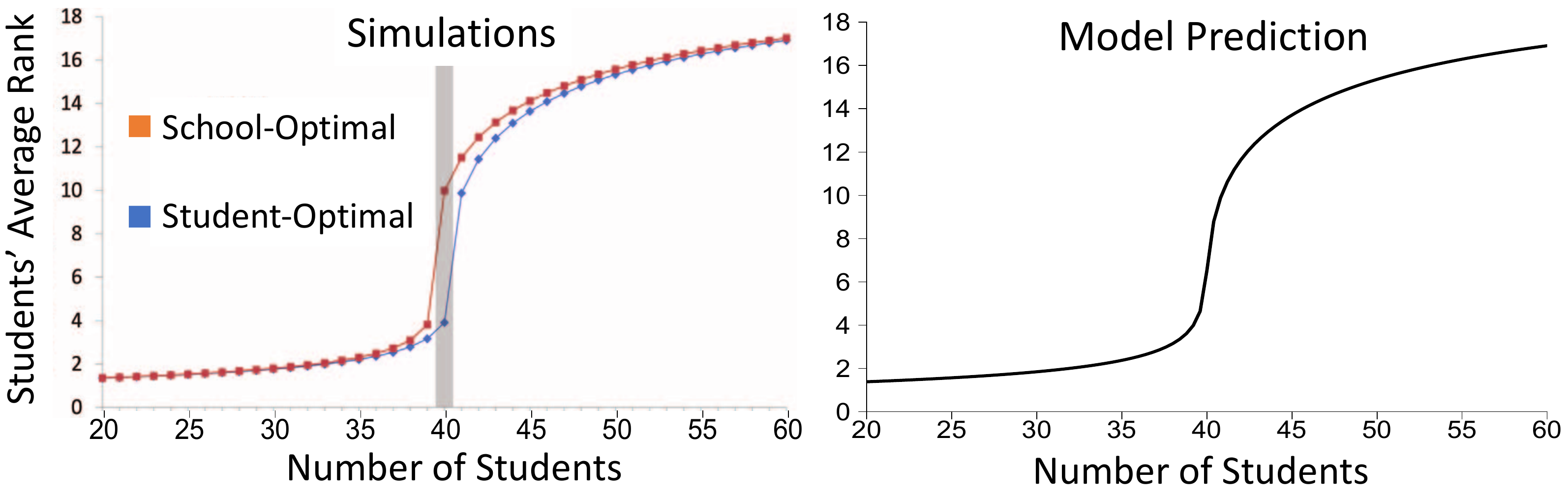}}
\caption{Students' average rank for their assigned school. In this market, there are 40 schools, each with a single seat. The number of students is given along the $x$ axis, and both student preferences and school priorities are drawn iid and uniformly at random. The left panel shows simulation results from \citet{ashlagi-kanoria-leshno_2017} demonstrating that (i) the difference between the school-optimal and student-optimal stable match is typically small, and (ii) in a balanced market (highlighted in gray), adding or removing one student has a dramatic effect. Our proposed model of stable matching generates a unique prediction (right), which closely matches the simulations and captures the dramatic effect of additional students in nearly-balanced markets.} \label{fig:akl2}
\end{figure}

\na{
\section{Exposition}

Consider a model of search. There are $m$ renters and $n > m$ available apartments. Each renter is shown one random available apartment. Apartments accept random renters. \na{Use nannies as motivation? Like school for small kids?}

Consider a market with $n$ symmetric schools, each with $C$ seats. There are $m$ students, each of which lists a random school. Schools rank students using school-specific lotteries. 

A canonical case to study is one where students and schools have iid uniform preferences and priorities, respectively.
}

\section{Model} \label{sec:stability}

There is a finite set of high schools $\H$. School $h \in \H$ has capacity $C_h \in \mathbb{N}$. We let $\H_0 = \H \cup \{\emptyset\}$ denote the set of schools along with the outside option of going unassigned, and define $C_\emptyset = \infty$. Let $\R$ be the set of complete orders over $\H_0$. We impose no restriction on the number of acceptable schools for each student (i.e. the number of schools preferred to the outside option $\emptyset$), and the order of schools ranked below $\emptyset$ will be irrelevant.

Students are characterized by their type $\theta = (\succ^\theta, {\bf p}^\theta)$, where $\succ^\theta \in \R$ indicates the student's preferences and ${\bf p}^\theta \in [0,1]^\H$ indicates the student's priority score at school $h$ (higher is better). We let $\Theta = \R \times [0,1]^\H$ denote the space of student types. Students are distributed according to a positive finite measure $\eta$ over $\Theta$. 
%In what follows, we place constraints on the measure $\eta$ that correspond to the assumption of strict priorities.

A fractional matching is a function $M$ mapping each $\theta \in \Theta$ to a probability distribution on $\H_0$. For each $h \in \H_0$ and $\theta \in \Theta$, the quantity $M_h(\theta)$ can be interpreted as the probability that a student of type $\theta$ is assigned to $h$.  Hereafter, we use ``matching" to mean a fractional matching, and denote the space of matchings by $\Matching$. 

We now define what it means for a matching to be {\em stable}. Our definition uses two auxiliary concepts, which are based on the perspective of individual agents. What matters to each student is the set of schools that admit them. What matters to a school is the set of students who are ``interested," meaning that they would attend if admitted. In our model, these are described by 
\begin{itemize}
\item An {\em admissions function} $A: [0,1] \rightarrow [0,1]^{\H_0}$. 
\item An {\em interest function} $\I : [0,1] \rightarrow \mathbb{R}_+^{\H_0}$.
\end{itemize}
Given $h \in \H$ and $p \in [0,1]$, $A_h(p)$ can be interpreted as the probability that a student with priority $p$ at $h$ will be admitted, while $\I_h(p)$ can be interpreted as the measure of interest in school $h$ from students whose priority at $h$ exceeds $p$. Given these interpretations, it is natural that $A_h$ should be increasing and $I_h$ should be decreasing. We let $\Admissions$ denote the set of componentwise weakly increasing functions from $[0,1]$ to $[0,1]^{\H_0}$, and let $\Interest$ denote the set of componentwise weakly decreasing functions from $[0,1]$ to $\mathbb{R}_+^{\H_0}$. 

We will define consistency conditions that link a matching $M$ to school interest $I$ and student admissions decisions $A$. Formally, we define maps $\mI : \Matching \rightarrow \interest$,  $\mA : \Interest \rightarrow \Admissions$ and $\M : \Admissions \rightarrow \Matching$, and define a stable matching as a fixed point of the composition of these maps. %(which depend on the market primitives ${\bf C}, \eta$, as well as the choice of vacancy function $\mV$)

Although our approach may seem unfamiliar, our definition subsumes existing ones. The function $\mI$ depends on the measure of student types $\eta$, and the function $\mA$ depends on a {\em vacancy function $\mV$}. Section \sref{sec:finite-special-case} shows that for a particular choice of $\eta$ and $\mV$, our definition of a stable matching coincides with the absence of blocking pairs in a finite market. Section \sref{sec:al-special-case} shows that when $\eta$ is changed to a continuous measure, any matching that is stable according to our definition is associated with a set of market-clearing cutoffs, and vice versa.

\subsection{Matching to Interest} \label{sec:matching-to-interest}

Given any matching $M \in \Matching$, define $\mI(M) \in \Interest$ to be the interest function $\I^M$ such that for each $h \in \H_0$ and $p \in [0,1]$, %\na{Could eliminate $h \succ^\theta \emptyset$ if we want, given that $M$ will ensure we never go below $\emptyset$}
\begin{equation} \I^M_h(p) =   \int {\bf 1}(p_h^\theta \ge p) (1 - \sum_{h' \succ^\theta h} M_{h'}^\theta) \, d \eta(\theta). \label{eq:continuum-interest} \end{equation}
Note that the sum in \eqref{eq:continuum-interest} gives the probability under matching $M$ that student type $\theta$ matches to a school preferred to $h$, so the interpretation of \eqref{eq:continuum-interest} is that students are interested in $h$ if they are not matched to any preferred school. The indicator ensures that the only students contributing to $\I^M_h(p)$ are those with priority above $p$ at $h$, allowing us to interpret $I_h^M(p)$ as the expected number of students with priority above $p$ who are interested in $h$.

%Because a student of type $\theta$ is interested in $h$ with probability $1 - \sum_{h' \succ^\theta h} M_{h'}^\theta$, this gives the measure of interest in $h$ from students with priority above $p$.

\remove{To gain intuition for the interest function $\I^M_h$, note that if $M$ is a deterministic matching induced by a cutoff vector $P$, then $1 - \sum_{h' \succ^\theta h} M_{h'}^\theta = \prod_{h' \succ^\theta h} {\bf 1}(p_{h'}^\theta \leq P_{h'})$ is an indicator that student $\theta$ is not admitted to any school preferred to $h$ and thus by \eqref{eq:demand}, $I^M_h(p)$ gives the demand at $h$ when $P_h = p$, holding fixed the cutoffs $\{P_{h'}\}_{h' \neq h}$.}

\subsection{Interest to Admissions}

The interest function $\I$ describes expected interest at each school $h \in \H$ and priority level $p \in [0,1]$. From this, we wish to determine an admissions function $A : [0,1] \rightarrow [0,1]^\H$, where $A_h(p)$ is interpreted as the probability that a student with priority $p$ at $h$ will be admitted to $h$ (equivalently, the probability that school $h$ has a final cutoff below $p$). We define $A$ using a {\em vacancy function} $\mV: \mathbb{R}_+ \times \mathbb{N} \rightarrow [0,1]$. Formally, we let $\mA(\I)$ be the admissions function $A^\I \in \Admissions$ that satisfies, for each $h \in \H$ and $p \in [0,1]$,
\begin{equation} A_h^\I(p) = \mV(\I_h(p), C_h). \label{eq:admissions}\end{equation}
We define $A_\emptyset^\I(p) = 1$ for all $p \in [0,1]$ (students are always admitted to the outside option).

The choice of vacancy function is an important feature of the model, and one of our key innovations. The quantity $\mV(\lambda,C)$ is interpreted as the probability that when {\em expected} interest is equal to $\lambda$, {\em realized} interest will be below $C$. Thus, if schools consider students in descending order of priority, $\mV(\I_h(p), C_h)$ gives the probability that school $h$ will still have at least one vacancy when it considers a student with priority $p$.

One natural choice of vacancy function is 
\begin{equation} \vdet(\lambda,C) = {\bf 1}(\lambda < C) \hspace{.3 in} \forall \lambda \in \mathbb{R}_+, C \in \mathbb{N}, \label{eq:vdet} \end{equation}
In other words, realized interest is {\em det}erministically equal to expected interest, and there is still a vacancy if and only if expected interest is below the school's capacity. This choice produces a deterministic prediction for each student type $\theta$ and each school $h$. We show in Sections \sref{sec:finite-special-case} and \sref{sec:al-special-case} that this choice of vacancy function allows us to recover the definition of stability in a finite market, as well as the definition used by \citet{azevedo-leshno_2016}.

In Section \sref{sec:poisson}, we use the following alternative choice of vacancy function, which assumes that when expected interest is equal to $\lambda$, realized interest follows a Poisson distribution with mean $\lambda$. \todo{Explain more?}
%study a continuum model with strict priorities, as defined in Definition \ref{def:strict-continuum}. In addition, we take capacities to be integral, and define 
\begin{equation} \vpois(\lambda,C) = \sum_{k = 0}^{C-1} \frac{e^{-\lambda} \lambda^k}{k!}. \label{eq:poisson} \end{equation}
Note that this choice produces admissions probabilities in $(0,1]$, reflecting the uncertainty facing participants in finite random matching markets.

%This choice is motivated by the thought that the number of students who are interested in $h$ and have priority above $p$ should follow a Poisson distribution with mean $\I_h(p)$. \todo{Explain motivation.}

%, and is therefore not well-suited to making predictions about random matching markets% where schools have small capacities.

%\na{Do we need any other conditions? If we want to include infinity, we could say $\mV(\lambda,\infty) = 1$ for all $\lambda \in \mathbb{R}_+$ -- this ensures that the outside option is always acceptable.}

\subsection{Admissions to Matching}

Recall that an admissions function $A$ describes the probability that a student of any given priority $p \in [0,1]$ will be admitted to each school. From this, we construct an associated fractional matching $\M(A) = M^A$ given by
\begin{equation} M_h^A(\theta)  =  A_h(p_h^\theta) \prod_{h' \succ^\theta h} (1 - A_{h'}(p_{h'}^\theta)). \label{eq:enroll-prob}  \end{equation}
This says that a student matches to $h$ if and only if she is admitted to $h$ and not to any preferred school. Note that this formula implicitly assumes independence of admissions outcomes across schools. A straightforward inductive argument implies that for any $A \in \Admissions$, $\theta \in \Theta$ and $h \in \H_0$,
\begin{equation} 1 - \sum_{h' \succ^{\theta} h} M_{h'}^A(\theta) = \prod_{h' \succ^{\theta} h} (1 - A_{h'}(p_{h'}^\theta)),  \label{eq:unravel} \end{equation}
with both sides interpreted to be $1$ if $h$ is the first choice of $\theta$. 

%\na{An equivalent definition is that $A = \mA(\mI(A))$, or $\I = \mI(\mA(\I))$. Could go back to defining stability with respect to one or the other of these functions.}

\subsection{Definition of Stability}

The two key parameters to our model are the measure of student types $\eta$ (which determines the function $\mI: \Matching \rightarrow \Interest$), and vacancy function $\mV$ (which determines the function $\mA : \Interest \rightarrow \Admissions$). For any $\eta$ and $\mV$, we employ the following definitions of stability. 

\na{Say ($\eta,\mV$)-stable, or just $\mV$-stable? $\eta$ included in $\mathcal{E}$, and I am already making the dependence on $C$ implicit. Could use $(\mE,\mV)$-stable.)}
\begin{definition} \label{def:stability} \text{ }

A matching $M \in \Matching$ is $(\eta, \mV)$-{\bf stable} if $M = \M(\mA(\mI(M)))$. 

An admissions function $A \in \Admissions$ is $(\eta, \mV)$-{\bf stable} if $A = \mA(\mI(\M(A)))$. 

An interest function $I \in \Interest$ is $(\eta, \mV)$-{\bf stable} if $\I = \mI(\M(\mA(\I)))$. 

An outcome $(M, \I, A) \in \Matching \times \Interest \times \Admissions$ is $(\eta,\mV)$-{\bf stable} if $M = \M(A), \I = \mI(M)$, and $A = \mA(I)$.
\end{definition}
Definition \ref{def:stability} implies that there is a one-to-one correspondence between stable matchings, stable admissions functions, stable interest functions, and stable outcomes. We include each of the definitions above because it is sometimes most convenient to work with stable matchings, and at other times simpler to work with stable interest functions or stable outcomes.

\section{Results} \label{sec:results}

Our definition of stability in Definition \ref{def:stability} may seem strange to those familiar with more traditional definitions based on the absence of blocking pairs, or cutoffs that clear the market. It more closely resembles fixed-point characterizations of stable matchings by \citet{adachi_2000}, \citet{fleiner_2003}, and \citet{echenique_2004}. Our first results bridge this divide by showing that when using the deterministic vacancy function $\vdet$ from \eqref{eq:vdet}, our definition encompasses more traditional definitions as special cases. Section \ref{sec:finite-special-case} shows that in finite markets, our definition corresponds to the absence of blocking pairs. Section \ref{sec:al-special-case} shows that in continuum markets, our definition is equivalent to one based on market-clearing cutoffs.

While these results both assume that the vacancy function $\mV$ is as given in \eqref{eq:vdet}, we proceed to establish that for {\em any} $\eta$ and $\mV$, several classic results hold: the set of stable matchings is a non-empty lattice, the extreme points of this lattice can be found using the deferred acceptance algorithm, and the rural hospital theorem applies. Finally, we prove that if $\eta$ has strict priorities and $\mV$ is strictly decreasing, there is a unique stable matching. %\todo{Mention that this is one of few uniqueness results, complementary to AL.}

 %(that is, markets in which $\S$ is generated by a random process). To our knowledge, all papers in this literature consider settings that can be captured by 
%the following procedure: fix a probability measure $\eta$ over $\Theta$, and let $\S$ be a set of $n$ types drawn iid from $\eta$. 
%The latter definition implies that if $\eta$ is a probability measure over $\Theta$ and $\S$ consists of iid draws from $\eta$, then $\S$ has strict priorities with probability one.
%When $\S$ is random, in order to ensure that $\S$ has strict priorities with probability one, it is necessary to impose mild constraints on the measure of student types.

%\begin{definition}[Finite Market with Strict Priorities] \label{def:strict-discrete} 
%A finite subset $\S \subset \Theta$ has {\bf strict priorities} if for each $h \in \H$ and $p \in [0,1]$, $|\S \cap \indif_h(p) | \leq 1$.
%%$|\indif_h(p) \cap S| \leq 1$.
%%The measure $\eta$ is {\bf finite} if there exists $S \subseteq \Theta$ with $|S| < \infty$ such that for each measurable $\tilde{\Theta} \subseteq \Theta$, $\eta(\tilde{\Theta}) = | S \cap \tilde{\Theta}|$. A finite market has {\bf strict priorities} if for each $h \in \H$ and $p \in [0,1]$, $\eta(\indif_h(p)) \leq 1$.
%\end{definition}

\subsection{Finite Markets: Stability = No Blocking Pairs} \label{sec:finite-special-case}

Traditionally, stable matching problems involve a finite set of students $\S \subset \Theta$. We adopt the standard assumption that $\S$ has ``strict priorities": no two students in $\S$ have identical priority at any school. 
%The assumption that $\S$ has strict priorities is necessary for Lemma \ref{lem:enrollment-finite} and Proposition \tref{prop:finite} to hold. If priorities are not strict, it is possible for an $(\eta^\S, \vdet)$-stable matching to be infeasible, and possible that no feasible $\S$-matching has no blocking pairs. 
Given $h \in \H$ and $p \in [0,1]$, define
\begin{equation} \indif_h(p) = \{ \theta : h \succ^\theta \emptyset, p_h^\theta = p\}\label{eq:indifference-set} \end{equation}
to be the set of student types that consider school $h$ acceptable and have priority $p$ at school $h$. 

\begin{definition}[Strict Priorities] \label{def:strict} \text{ } 

A finite subset $\S \subset \Theta$ has {\bf strict priorities} if for each $h \in \H$ and $p \in [0,1]$, $|\S \cap \indif_h(p) | \leq 1$.

A positive measure $\eta$ on $\Theta$ has {\bf strict priorities} if for each $h \in \H$ and $p \in [0,1]$, $\eta(\indif_h(p)) = 0$.
\end{definition} 
The second part of this definition is motivated by the study of {\em random} finite matching markets, where $\S$ is generated by drawing student types iid from some measure $\eta$ over $\Theta$. In this case, the condition above ensures that $\S$ has strict priorities with probability one.\footnote{The assumption that there are no ties is essential to many of our results. This is not an artifact of our definitions or proof techniques, but rather reflects fundamental challenges to defining stable matchings with indifferences.}

We now give a version of the traditional definition of stability based on the absence of blocking pairs. We refer to this concept as ``no blocking pairs" to distinguish it from the definition of stability in Definition \ref{def:stability}.

\begin{definition}[No Blocking Pairs] \label{def:finite-stability}
Given any finite set $\S \subset \Theta$, an $\S$-matching is a function $\mu: \S \to \H_0$. 
An $\S$-matching $\mu$ is {\bf feasible} if for each $h \in \H_0$, 
\begin{equation} | \{\theta \in \S : \mu(\theta) = h\}| \leq C_h. \end{equation}
An $\S$-matching $\mu$ {\bf has no blocking pairs} if it is feasible, and for each $\theta' \in S$ and each $h \in \H_0$ such that $h \succ^{\theta'} \mu(\theta')$, 
\begin{equation} | \{\theta \in \S : \mu(\theta) = h, p_h^\theta > p_{h}^{\theta'}\}| = C_h. \label{eq:no-blocking} \end{equation}
\end{definition}

Definition \ref{def:finite-stability} states that a feasible $\S$-matching has no blocking pairs if for each student $\theta' \in \S$, each school that $\theta'$ prefers to its assignment is filled with higher-priority students. Note that this implies individual rationality: because $C_{\emptyset} = \infty$, the outside option is never filled to capacity. Therefore, if $\mu$ has no blocking pairs, then it does not assign any student to a school that she considers inferior to the outside option. 

Our first result is to show that our definition of stability corresponds with the traditional definition based on the absence of blocking pairs. To state this result formally, we note that each finite set $\S \subset \Theta$ is naturally associated with an associated counting measure $\eta^\S$ over $\Theta$, defined by
\begin{align} 
\eta^\S(\tilde{\Theta}) & = | \tilde{\Theta} \cap \S| & \forall \tilde{\Theta} \subseteq \Theta. \label{eq:etaS}%\\
\end{align}

Similarly, there is a natural correspondence between $\S$-matchings (which define an assignment only for student types in $\S$) and deterministic matchings (which define an assignment for all types in $\Theta$). Any deterministic matching $M$ naturally defines an $\S$-matching $\mu^M$: for each $\theta \in \S$, let
\begin{equation}\mu^M(\theta) = h \Leftrightarrow M_h(\theta) = 1.\end{equation}
%define $\mu^M(\theta) = h$ for the one $h \in \H_0$ such that $M_h(\theta) = 1$. 
Similarly, each $\S$-matching $\mu$ naturally induces a deterministic matching $M^\mu$ as follows. Define the admissions outcome $A^\mu$ by
\begin{equation} A_h^\mu(p) = {\bf 1}(| \{\theta \in \S : p_h^\theta > p,  \mu(\theta) = h\}| < C_h),\label{eq:amu} \end{equation} %= \vdet(| \{\theta \in \S : \mu(\theta) = h, p_h^\theta > p\}|,C_h)
and define $M^\mu = \M(A^\mu)$. In other words, \eqref{eq:amu} says that student $\theta \in \Theta$ is admitted to $h$ if there are fewer than $C_h$ higher-priority students from $\S$ matched to $h$ under $\mu$, and $M^\mu(\theta)$ is the matching that results when each student type $\theta$ is assigned to its most-preferred school among those where it is admitted.

The following result says that if priorities are strict, then the functions $M \to \mu^M$ and $\mu \to M^\mu$ define a bijection between the set of $(\eta^\S, \vdet)$-stable matchings, and the set of $\S$-matchings with no blocking pairs. The proof of this result is deferred to Appendix \sref{app:prop-finite}.

\begin{proposition}[No Blocking Pairs Corresponds to a Stable Matching] \label{prop:finite}
Let $\S$ be a finite subset of $\Theta$ with strict priorities. If $M$ is a $(\eta^\S, \vdet)$-stable matching, then $\mu^M$ has no blocking pairs. If $\mu$ is an $\S$-matching with no blocking pairs, then $M^\mu$ is $(\eta^\S, \vdet)$-stable, and $\mu^{M^{\mu}} = \mu$.
%\todo{coincides with $\mu$ on $\S$: for each $\theta \in \S$ and $h \in \H_0$, $M_h^\mu(\theta) = {\bf 1}(\mu(\theta) = h)$.}
\end{proposition}

\subsection{Continuum Markets: Stability = Market-Clearing Cutoffs} \label{sec:al-special-case}

\citet{azevedo-leshno_2016} provide a continuum model in which a market is described by a positive measure $\eta$ over $\Theta$ and a stable matching is described by a vector of priority cutoffs $P \in [0,1]^\H$. Students are admitted to school $h$ if and only if their priority at $h$ exceeds its cutoff $P_h$. They define demand for school $h$ at cutoffs $P$ by\footnote{An astute and informed reader might notice that our choice of $A^P$ assumes that student types $\theta$ with $p_h^\theta = P_h$ are not admitted to $h$, whereas \citet{azevedo-leshno_2016} assume that they are admitted. Because $\eta$ is a continuum measure with strict priorities in both cases, this distinction is of consequence only to sets of $\eta$-measure zero.}
\begin{align}
\demand_h(P) = \int {\bf 1}(p_h^\theta > P_h) \prod_{h' \succ^\theta h} {\bf 1}(p_{h'}^\theta \leq P_{h'}) d\eta(\theta). \label{eq:demand}
\end{align}
That is, demand for $h$ at cutoffs $P$ is equal to the measure of students who are admitted to $h$ and are not admitted to any school that they prefer to $h$. Each cutoff vector is naturally associated with a deterministic matching in which students attend the school that they demand. The definition of stability used by \citet{azevedo-leshno_2016} is that the cutoff vector should clear the market.
\begin{definition}
A cutoff $P \in [0,1]^\H$ is $\eta$-{\bf market-clearing} if $\demand_h(P) \leq C_h$ for all $h \in \H$, with equality if $P_h > 0$. \label{def:market-clearing}
\end{definition}

In this section, we show that in continuum markets, a cutoff vector $P$ is $\eta$-market-clearing if and only if a corresponding interest function is $(\eta, \vdet)$-stable. To formalize this claim, we first define a natural associate between cutoff vectors and interest functions. Each cutoff vector $P$ is naturally associated with an interest function $I^P$ defined for each $h \in \H$ and $p \in [0,1]$ by
\begin{equation} I^P_h(p)  = \int {\bf 1}(p_h^\theta > p) \prod_{h' \succ^\theta h} {\bf 1}(p_{h'}^\theta \leq P_{h'}) d\eta(\theta). \label{eq:iP}\end{equation}
That is, students contribute to this quantity if they have priority above $p$ at $h$ and do not ``clear the cutoff" at any school that they prefer. Note that when $p = P_h$, we have $I^P_h(P_h) = D_h(P)$.
 
Conversely, from any interest function $I \in \Interest$, we can define the associated cutoffs  $\P(\I) = \{ \P_h(\I) \}_{h\in \H} \in [0,1]^\H$ by
 \begin{equation} \P_h(\I) = \inf \{p \geq 0 : \I_h(p) < C_h\}. \label{eq:cutoffs} \end{equation}

Equation \eqref{eq:iP} defines a mapping from cutoffs to interest functions, while \eqref{eq:cutoffs} defines mapping from interest functions to cutoffs. It turns out that these mappings take stable interest functions to market-clearing cutoffs, and vice versa.\footnote{We briefly comment on a subtlety that explains why Proposition \tref{prop:al} is stated in terms of the interest function $I^P$ rather than the admissions function $A^P$ or the matching $M^P$. In general, multiple market-clearing cutoffs may correspond to the same stable matching (up to a set of measure zero). For example, suppose that there is a single school $h$ with capacity $C$, and that the total measure of students is $\eta(\Theta) = 2C$, with priorities uniformly distributed on $(0,1/3) \cup (2/3,1)$. Then any cutoff $P \in [1/3,2/3]$ clears the market. Our definition of stability eliminates this redundancy: the unique $(\eta,\vdet)$-stable matching corresponds to a cutoff of $2/3$ and leaves students in $[0,2/3]$ unassigned. Thus, if $P \in [1/3,2/3)$, $P$ clears the market but $M^P$ is not $(\eta,\vdet)$-stable. By contrast, for any $P$, $I^P(p) =  \eta(\{\Theta : p_h^\theta > p\})$ is a stable interest function.
}
 \begin{proposition}[Market-Clearing Cutoffs Correspond to Stable Interest Functions] \text{ } \\
\label{prop:al}
Let $\eta$ have strict priorities. If $P$ is $\eta$-market-clearing, then $I^P$ is $(\eta,\vdet)$-stable. If $I$ is $(\eta,\vdet)$-stable, then $\P(I)$ is $\eta$-market-clearing, and $I = I^{\P(I)}$.
 %$\tilde{M} = \M(\mA(\mI(M^P)))$ is $(\eta,\vdet)$-stable and results in the same outcome as $P$ for almost all $\theta$. That is, for each $h \in \H$, 
%$\[ \eta(\{\theta : M^P_h(\theta) \neq M_h^{\P(I^P)}(\theta) \}) = 0.\]
\end{proposition}

We prove this result in Appendix \sref{app:prop-al}.

\subsection{Existence and Lattice Structure}  \label{sec:existence}

Having established that when $\mV = \vdet$, our definition of stability nests existing definitions, we now prove results for general type measures $\eta$ and vacancy functions $\mV$. The first of these results shows that stable matchings always exist and form a lattice. To state this result, we define the following partial orders:
\begin{itemize}
\item $M \succeq^\Matching \tilde{M}$ if for each $h \in \H_0$ and $\theta \in \Theta$, 
\[\sum_{h' \succeq^\theta h} M_{h'}(\theta) \geq \sum_{h' \succeq^\theta h} \tilde{M}_{h'}(\theta).\]
That is, $M \succeq^\Matching \tilde{M}$ if each student prefers $M$ to $\tilde{M}$ in the sense of first-order stochastic dominance.
%each student's assignment under $M$ stochastically dominates their assignment under $\tilde{M}$. 
\item $A \succeq^\Admissions \tilde{A}$ if for each $h \in \H_0$ and $p \in [0,1]$, $A_{h}(p) \geq \tilde{A}_h(p)$. \\
That is, $A \succeq^\Admissions \tilde{A}$ if admissions probabilities are uniformly higher under $A$.
\item $\I \succeq^\Interest \tilde{\I}$ if for each $h \in \H_0$ and $p \in [0,1]$, $\I_{h}(p) \geq \tilde{\I}_h(p)$. \\
That is, $\I \succeq^\Interest \tilde{\I}$ if each school receives more interest at every cutoff under $\I$.
\item $(M,A,\I) \succeq (\tilde{M}, \tilde{A},\tilde{\I})$ if $M \succeq^\Matching \tilde{M}$, $A \succeq^\Admissions \tilde{A}$, and $\tilde{\I} \succeq^\Interest \I$.
\end{itemize}
%That is, $M \succeq^\Matching \tilde{M}$ 

\begin{theorem}[Existence and Lattice Structure] \label{thm:existence} 
If the vacancy function $\mV$ is weakly decreasing in its first argument, then for any $(\H, {\bf C}, \eta)$, the set of $(\eta,\mV)$-stable outcomes is non-empty, and forms complete lattice with partial order $\succeq$.
\end{theorem}

\begin{proof}[Proof of Theorem \tref{thm:existence}]
Define the function $\xi : \Matching \rightarrow \Matching$ by 
\begin{equation} \xi(M) =\M(\mA(\mI(M))). \end{equation}
Note that
\begin{itemize}
\item By \eqref{eq:continuum-interest}, $\tilde{M} \succeq^\Matching M$ implies $\mI(M) \succeq^\Interest \mI(\tilde{M})$. 
\item By \eqref{eq:admissions} and monotonicity of $\mV$, $\I \succeq^\Interest \tilde{\I}$ implies $\mA(\tilde{\I}) \succeq^\Admissions \mA(\I)$.
\item By \eqref{eq:enroll-prob}, $\tilde{A} \succeq^{\Admissions} A$ implies $\M(\tilde{A}) \succeq^{\Matching} \M(A)$.
\end{itemize}
From this, we draw two conclusions. First, if $(M,\I,A)$ and $(\tilde{M},\tilde{\I}, \tilde{A})$ are stable outcomes, then $(M,\I,A) \succeq (\tilde{M},\tilde{\I}, \tilde{A})$ if and only if $M \succeq^\Matching \tilde{M}$. Second, the function $\xi$ is an order preserving function, so Tarski's fixed point theorem implies that the set of fixed points of $\xi$ (that is, the set of stable matchings) forms a complete lattice with respect to $\succeq^\Matching$ (and in particular is non-empty). 
\end{proof}

\subsection{Deferred Acceptance Algorithm} \label{sec:da-general} 

Theorem \tref{thm:existence} establishes the existence of stable outcomes, but does not address how to find them. However, the proof suggests a natural procedure: start from a matching $M$ and repeatedly apply the function $\xi$ defined by $\xi(M) = \M(\mA(\mI(M)))$. If one starts from the matching $\overline{M}$ which assigns each student to her most preferred school, then this procedure corresponds to the student-proposing deferred acceptance algorithm, and converges to the student-optimal stable matching. To see that it converges, note that $\xi(\overline{M}) \preceq^\Matching \overline{M}$, from which the fact that $\xi$ is order-preserving implies that the sequence $\{\xi^k(\overline{M}) \}_{k = 0}^\infty$ is decreasing. Therefore, it converges by completeness of $\Matching$. Conversely, repeatedly applying $\xi$ from the student-pessimal matching $\underline{M}$ (defined by $\underline{M}_\emptyset(\theta) = 1$ for all $\theta$) produces an increasing sequence of matchings that converges to the school-optimal stable matching. 

Although convergence is guaranteed, in general it does not occur in finitely many steps. In examples that we have tried, convergence happens quickly enough that this algorithm can be applied fruitfully. The main practical challenge is computing $\mI(M)$, which requires taking an integral over student types. Although this may be challenging for arbitrary measures $\eta$, it is tractable for many cases of interest.

\begin{figure}
\captionsetup{width=\textwidth}
\includegraphics[width = .8\textwidth]{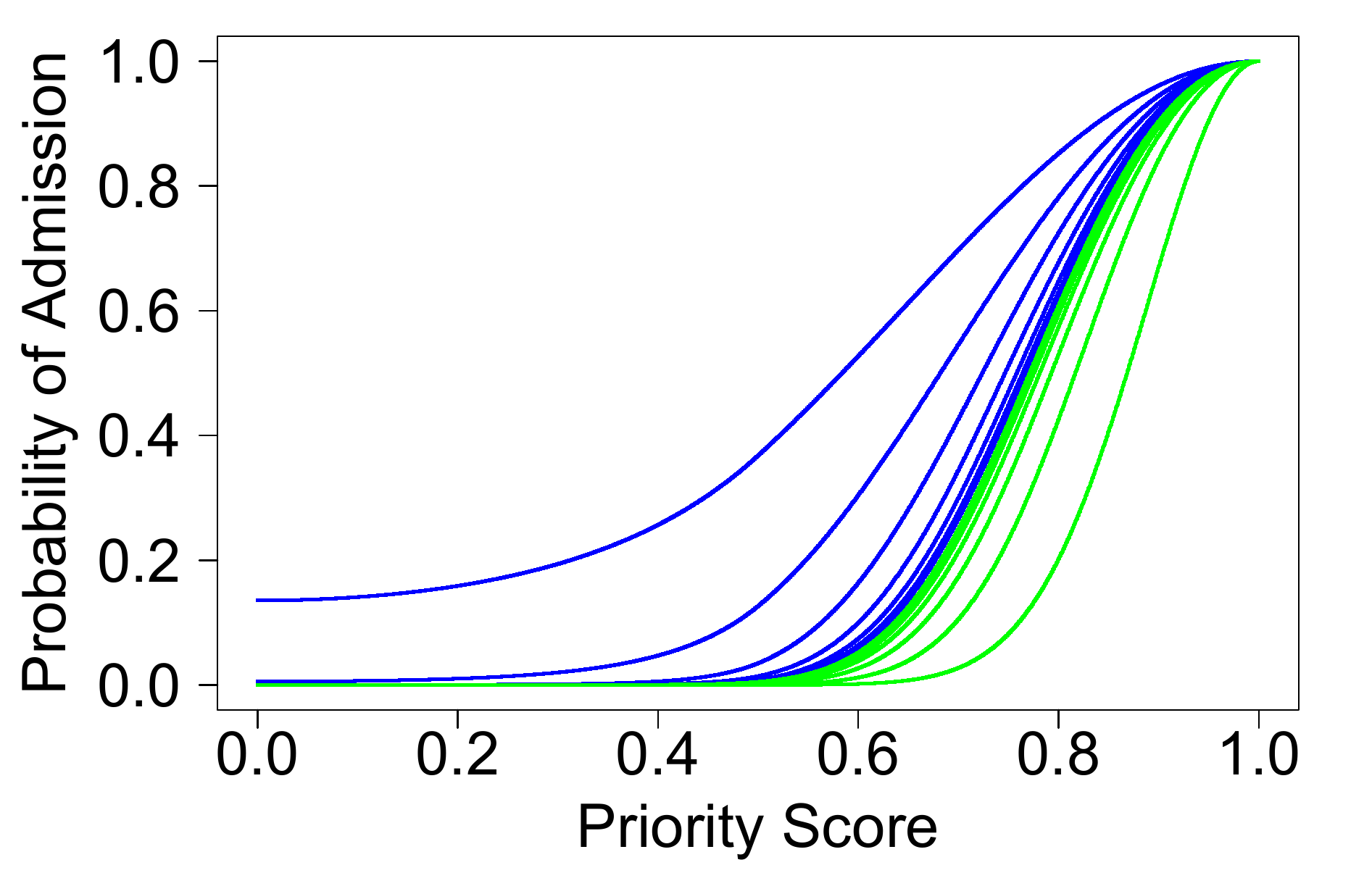}
\caption{Finding a stable admissions function for the example from \citet{azevedo-leshno_2016}, using vacancy function $\vpois$ defined in \eqref{eq:poisson}. Because schools are symmetric, $A_h = A_{h'}$ for $h, h' \in \H$. We iterate the map $A \rightarrow  \mA(\mI(\M(A)))$. Starting from $A(p) = 1$ results in a decreasing sequence of admissions functions (in blue) converging to the student-optimal stable admissions function. Starting from $A(p) = 0$ results in an increasing sequence of admissions functions (in green) converging to the school-optimal stable admissions function. Because $\vpois$ is strictly decreasing, Theorem \tref{thm:uniqueness} guarantees that these admissions functions coincide. Notice that the admissions function is the CDF of the cutoff distribution whose $5^{th}$ and $95^{th}$ percentiles are shown in Figure \tref{fig:azevedo-leshno}.} \label{fig:da}
\end{figure}

Because of the correspondence between stable matchings, stable interest functions, and stable admissions functions, it is also possible to apply an analogous iterative process using the admissions function as the primitive of interest. In that case, one could start from $A_h(p) = 1$ for all $h$ and $p$ (resulting in convergence to the student-optimal stable matching) or $A_h(p) = 0$ for all $h$ and $p$ (resulting in convergence to the student-pessimal stable matching). Figure \fref{fig:da} shows the sequence of admissions functions that result when this iterative process is applied to an example from \citet{azevedo-leshno_2016}.

%It is also bounded, as \eqref{eq:fullconsistency} implies that for any $\I \in \Interest^\T$, and any $\tau$ and $s$, we have $(\xi \I)^\tau(s) \leq \rho  \E_{\ell \sim \L}[\ell]/\Dtau$. Therefore, this sequence converges (though not in finitely many steps). 
%In fact, the process of repeatedly applying $\xi$ to $\underline{\I}$ corresponds closely to the well-known student-proposing deferred acceptance algorithm. At the point $\underline{\I}$, schools have received no interest from students, so any student can be admitted to any school. The point $\xi \underline{\I}$ corresponds to the outcome after all students have expressed interest in their first choice school. Some of these students are not accepted, and express interest in their second school; the point $\xi^2 \underline{\I}$ corresponds to the outcome after this occurs. Although the convergence of $\{\xi^k \underline{\I}\}$ and $\{\xi^k \overline{\I}\}$ does not occur in finitely many steps, in examples that we have tried, convergence happens quickly enough that this algorithm can be applied fruitfully. The main practical challenge is computing $\xi \I$, which requires taking an expectation over student types $\theta \sim \eta$. Although this may be challenging for arbitrary distributions $\eta$, it is tractable for many cases of interest.

%\na{Does applying $\xi$ to $\underline{M}$ correspond to the school-proposing algorithm?}

\subsection{Rural Hospital Theorem}

We now establish a ``rural hospital theorem," which states that for any two stable matchings, each student's probability of assignment and each school's measure of matched students is identical. This is a generalization of the corresponding result for finite markets, proved by \citet{mcvitie-wilson_1970} and \citet{roth_1986}.

%\na{This is where we need appropriate definition of $\filledseats$ from $\mV$, or conditions on $\mV$ given in Assumption \ref{as:vacancy}.}
%\begin{lemma} \label{lem:matchconsistency}
%For any stable outcome $(M,\I,A)$, $\filledseats(\I) = \matchedstudents(A)$.
%\end{lemma}

\begin{theorem}[Rural Hospital Theorem] \label{thm:rht}
If $\eta$ has strict priorities and $\mV$ is weakly decreasing in its first argument, then the set of matched agents is identical across stable outcomes: if $(M, \I, A)$ and $(\tilde{M}, \tilde{\I}, \tilde{A})$ are $(\eta,\mV)$-stable outcomes, then for each $h \in \H_0$,
\begin{equation} \int M_h(\theta) d\eta(\theta) = \int \tilde{M}_h(\theta) d\eta(\theta),\label{eq:school-rht} \end{equation}
and for each $\theta \in \Theta$ outside of a set of $\eta$-measure zero, 
%\[M_\emptyset^\theta = \tilde{M}_\emptyset^\theta,\]
\begin{equation} \sum_{h \succ^\theta \emptyset} M_h(\theta) = \sum_{h \succ^\theta \emptyset} \tilde{M}_h(\theta). \label{eq:student-rht} \end{equation}
\end{theorem}

In contrast to the existence result in Theorem \tref{thm:existence}, Theorem \tref{thm:rht} requires an assumption of strict priorities. This assumption is essential for the result to hold.\footnote{To see that the conclusion of Theorem \tref{thm:rht} may fail to hold if $\eta$ does not have strict priorities, consider an example with two schools, $A$ and $B$, each with a single seat. The measure $\eta$ corresponds to a finite market with three students, $x,y,z$. Students $x$ and $y$ prefer $A$ to $B$, while student $z$ prefers $B$ to $A$. Student $x$ has priority $1/4$ at school $A$ and $3/4$ at school $B$. Student $y$ has priority $1/4$ at school $A$ and $2/4$ at school $B$. Student $z$ has priority $3/4$ at school $A$ and $1/4$ at school $B$.

We claim that there are two $(\eta,\vdet)$-stable matchings, and that $y$ is assigned in one and unassigned in the other. In the school-optimal stable matching, $x$ goes to $B$, $z$ goes to $A$, and $y$ is unassigned.  In the student-optimal stable matching, $x$ and $y$ go to $A$, and $z$ goes to $B$.  Note that the student-optimal stable matching is infeasible (two students are assigned to $A$). This illustrates that our definition of stability (which was intended for markets with strict priorities) does not enforce capacity constraints in markets with ties.}
The failure of the rural hospital theorem when there are ties in priority is not specific to our definition of stability: in finite markets with indifferences, it is known that strongly stable matchings may not exist \citep{irving_1994}, and weakly stable ones may not satisfy the rural hospital theorem \citep{manlove_1999}.

\subsection{Uniqueness}

\na{Continuum measure $\rightarrow$ continuous measure.}

Finally, we establish conditions under which there is a unique stable outcome. 

%The following result implies that the Poisson model introduced in Section \sref{sec:poisson} has a unique stable outcome.

\begin{theorem}[Uniqueness] \label{thm:uniqueness}
If $\eta$ has strict priorities and $\mV$ is strictly decreasing in its first argument, then there is a unique $(\eta,\mV)$-stable outcome.
\end{theorem}

The intuition underlying this result is as follows. By Theorem \ref{thm:existence}, there are student-optimal and student-pessimal stable admissions functions $A$ and $\tilde{A}$, with $A \succeq \tilde{A}$. It follows that all students will be weakly more likely to match under $A$. If $\mV$ is strictly decreasing, then $A \succ \tilde{A}$ implies that some students will be strictly more likely to match under $\tilde{A}$. This contradicts the rural hospital theorem, implying that we  must have $A = \tilde{A}$. The complete proof is provided in Appendix \ref{app:main-results}. 

%either of the following holds: 
%\begin{enumerate}[label=(\roman*)]
%\item  $\mV(\cdot, C)$ is strictly decreasing for each $C \in \mathbb{N}$, or 
%\item $\eta$ has full support.
%\end{enumerate}
%\na{Could we not have different $\I$'s at the high end (i.e. if there are no students above that point)? No -- then we have defined $\I$ to be zero.}

If $\mV$ is only weakly decreasing, it is possible that there are multiple stable matchings that lead to different outcomes for a positive $\eta$-measure of students: see \citet{azevedo-leshno_2016} for an example with $\mV = \vdet$. However, their Theorem 1 shows that even in this case, there is typically a unique stable matching: this holds if $\eta$ has full support, or for a generic set of school capacities.

\section{A New Vacancy Function} \label{sec:poisson}

% \na{Finally, Section \sref{sec:poisson} introduces a novel vacancy function $\mV$ which results in fractional stable matchings (unlike the deterministic matchings that result from the vacancy function used in Sections \sref{sec:finite} and \sref{sec:al}). These matchings can be interpreted as  predictions about assignment probabilities in markets where student types are drawn randomly.}

In this section, we study model predictions when using the vacancy function $\vpois$ given in \eqref{eq:poisson}, which assumes that realized interest follows a Poisson distribution. 

To motivate this assumption, we revisit our discussion of the introduction. Suppose that schools post deterministic cutoff scores given by the vector $P \in [0,1]^\H$ and we generate a set $\S$ of $n$ students by sampling iid from a probability measure $\tilde{\eta}$ on $\Theta$. Then for each $h \in \H$ and each $p \in [0,1]$, the set of students who are in $\S$, interested in $h$, and have priority above $p$ at $h$ will follow a Binomial distribution with parameters $n$ and $I^P_h(p)/n$ (where we define $I^P$ as in \eqref{eq:iP}, with $\eta = n \tilde{\eta}$). If $n$ is modestly large, this will be well-approximated by a Poisson with parameter $I^P_h(p)$.

% then realized demand at each school $h$ will follow a Binomial distribution with parameters $n$ and $D_h(P)/n$, which is well approximated by a Poisson distribution with parameter $D_h(p)$. Similarly, for {\em any} potential cutoff $p \in [0,1]$, realized interest at $h$ from students with priority above $p$ will follow , and be . 
% 
% This motivates the

\na{ 

\subsection{Asymptotic Equivalence}

Given a measure $\eta$, define the metric $d_\eta$ on the space of admissions functions $\Admissions$ by
\begin{equation} d_\eta(A,A') =  \frac{1}{2}\int_\Theta \norm{M_h^A(\theta) - M_h^{A'}(\theta) } d\eta(\theta). \label{eq:d-metric}\end{equation}
Note that $\frac{1}{2}\norm{M_h^A(\theta) - M_h^{A'}(\theta) }$ is the total variation distance between $M^A(\theta)$ and $M^{A'}(\theta)$, so $d_\eta(A,A')$ can be interpreted as the expected number of students whose assignment changes from $A$ to $A'$.

\na{Discuss scaling up capacity of each school.}

\begin{proposition} \label{prop:convergence}
Fix $\H$ and ${\bf C} \in \mathbb{N}^\H$. Let $\eta$ be a continuous measure with strict priorities and full support, and let $\mathcal{E} = (\H, {\bf C}, \eta)$. Define a sequence of markets $\mathcal{E}^m = (\H, {\bf C}^m, \eta^m)$, with $C_h^m= m C_h$ for all $h \in \H$ and $\eta^m(\tilde{\Theta}) = m \cdot \eta(\tilde{\Theta})$ for all measurable $\tilde{\Theta} \subseteq \Theta$. 

Then there exists a unique $(\mE,\vdet)$-stable admissions function $\tilde{A}^{det}$, a unique $(\mE^m,\vpois)$-stable admissions function $\tilde{A}^m$ for each $m$, and $d_\eta(\tilde{A}^m,\tilde{A}^{det}) \rightarrow 0$ as $m \rightarrow \infty$.
%,\footnote{\na{Unique by Theorem \ref{thm:unique}}} and let \[\tilde{\Theta} = \{ \theta :  M^m(\theta) \rightarrow M(\theta) \}\] be the set of student types for which the match probabilities $M^m$ converge to $M$. 
%and let $A^m$ be a stable admissions function in market $M^m$.  \na{Do we need uniqueness conditions?} If there is a unique stable outcome $(A^{det}, \I^{det})$ in $M_1$ when $\mV(\lambda,C) = {\bf 1}(\lambda < C)$
%
%, then $A^m \rightarrow A^{det}$ as $m \rightarrow \infty$. \na{Say what this convergence means -- presumably, that $A_h^m(p) \rightarrow A_h^{det}(p)$ for all $p$ outside a set of measure zero.}
%\P(\I_n,{\bf C}_n) \rightarrow P$.
%Claim that distribution of cutoffs converges in probability to stable cutoffs in Azevedo Leshno. But which one? \todo{What happens if you feed in a distribution that has multiple stable matchings to my original model? Could it be that you don't converge, but might oscillate? So that any convergent subsequence converges to a stable matching?}
\end{proposition}

\begin{lemma} \label{lem:uniform-convergence-helpful}
Let $\mathcal{X}$ be a compact metric space, and let $f_n$ and $f$ be continuous functions from $\mathcal{X}$ to $\mathbb{R}_+$. Suppose that $f_n$ converges uniformly to $f$, that each $f_n$ has a unique root $x_n$, and that $f$ has a unique root $x$. Then $x_n \rightarrow x$.\end{lemma}

\na{An example showing that compactness is necessary. Let $X = \mathbb{R}_+$, and let $f(x) = 0$ and $f(x) = \min(x,1/x)$ for $x > 0$. Let $f_n(x) = f(x)$ for $x \in [1/n,n]$ and $x > n+1$, with $f_n(x) = 1/n$ for $x \in [0,1/n]$ and $f_n$ dipping down to graze zero between $n$ and $n+1$. Then $f_n$ converges uniformly to $f$ (the maximum distance between the two is $1/n$), and $f$ and $f_n$ have unique roots, but the roots of $f_n$ do not converge to the root of $f$.}

\na{An example showing that uniform convergence is necessary. Let $X = [0,2]$ (which is compact), with $f(x) = \min(x,1)$. Let $f_n$ be identical to $f$ on $[1/n,1]$, with $f_n(x) = 1/n$ on $[0,1/n]$. For $x \in [1,2]$, let $f_n$ be identically 1, except for a quick dip down to zero at $x = 1+1/n$. Then $f_n$ converges to $f$ (though not uniformly), and the roots of $f_n$ converge to $1$, not the root of $f$.} 
\begin{proof}[Proof of Lemma \ref{lem:uniform-convergence-helpful}]
Seeking a contradiction, suppose that the $x_n$ do not converge to $x$. That is, for any $\varepsilon > 0$, the cardinality of the set $\{n : d(x_n, x) > \varepsilon\}$ is infinite. By compactness, among the points in $\{x_n\}$ not within a distance $\varepsilon$ of $x$, there is a convergent subsequence; say that this subsequence converges to a point $y$. Furthermore, $0 = f_{n_k}(x_{n_k}) \rightarrow f(y)$, where we have used the fact that the $f_n$ converge uniformly to $f$ and that $x_{n_k}$ converge to $y$. But $y \neq x$ (in particular, $d(x,y) > \varepsilon$), so $f$ must have two roots, contradicting our initial assumption.
\todo{Could write a little more cleanly.}
\end{proof}

\begin{lemma}
If $\eta$ has full support, the function $d_\eta$ given in \eqref{eq:d-metric} defines a metric on the set of component-wise left-continuous admissions functions. \label{lem:metric} 
%\na{For discrete case, we want admissions functions to be right continuous (else, they admit everyone below person $C_h$). For continuous case, doesn't really matter. Note that if $\eta$ has full support but not strict priorities, $\I$ is right continuous but not continuous, and thus $A^{pois}$ is right continuous, not left continuous. So maybe be add assumption of strict priorities for $\eta$? Meanwhile with $\vdet$, we end up with left continuous admissions function. Really, it would be nice to use strict inequality $\vdet$ for discrete case, and weak inequality for continuous case.}
\end{lemma}
\begin{proof}[Proof of Lemma \ref{lem:metric}]
Symmetry and the triangle inequality are immediate. It remains to show that if $A \neq A'$ are left-continuous admissions functions, then $d_\eta(A,A') \neq 0$. Let $h \in \H$ and $\tilde{p} \in [0,1]$ be such that $A_h(\tilde{p}) \neq A'_h(\tilde{p})$, and without loss of generality assume $A_h(\tilde{p}) > A'_h(\tilde{p})$. Then it follows by left continuity of $A_h$ that there is some interval to the left of $\tilde{p}$ such that $A_h(p) > A'_h(p)$ for all $p$ in this interval. Thus, for any student types $\theta$ that rank $h$ first and have priority $p_h^\theta$ in this interval, $||M^A(\theta) - M^{A'}(\theta)||_1 > 0$. Because $\eta$ has full support, the $\eta$-measure of this set of student types is positive, implying that $d_\eta(A,A') > 0$.

\na{If want to show that it is compact, may use that compact is equivalent to complete and totally bounded.}
\end{proof}

\todo{Right now lots of definitions, no explanation or motivation. Make it easier to parse.}
\begin{lemma}
Fix $\H$, $C \in \mathbb{N}^\H$ and let $\eta$ be a measure with strict priorities and full support. Define $\mI$ as in \eqref{eq:continuum-interest}. \\
Define $\mA^{det}: \Interest \rightarrow \Admissions$ to be the admissions function $A$ satisfying $A_h(p) = \vdet(\I_h(p),C_h)$. \\ For $m \in \mathbb{N}$, define $\mA^{m}: \Interest \rightarrow \Admissions$ to be the admissions function $A$ satisfying $A_h(p) = \vpois(m\I_h(p),mC_h)$. \\
\na{Define $\xi : \Admissions \rightarrow \Admissions$ by $\xi(A) = \mA^{det}(\mI(\M(A)))$, and define $\xi^m : \Admissions \rightarrow \Admissions$ by $\xi^m(A) = \mA^m(\mI(\M(A)))$.}
Define $\Delta, \Delta^m : \Admissions \rightarrow \mathbb{R}_+$ by $\Delta(A) = d_\eta(A, \xi(A))$, and $\Delta^m(A) = d_\eta(A, \xi^m(A))$. \\
Then $\Delta^m$ converges uniformly to $\Delta$.
\label{lem:uniform-convergence-holds}
\end{lemma}

\begin{proof}[Proof of Lemma \ref{lem:uniform-convergence-holds}]
We must show that as $m \rightarrow \infty$,
\[\sup_{A \in \Admissions} \abs{\Delta^m(A) - \Delta(A)} \rightarrow 0.\]
To reduce notation, we fix $A \in \Admissions$ and define $A^m = \xi^m(A)$, $A^{det} = \xi(A)$, $M^m = \M(A^m)$ and $M^{det} = \M(A^{det})$. Note that because $d_\eta$ satisfies the triangle inequality, 
\begin{align*}
\abs{\Delta^m(A) - \Delta(A)} &  = \abs{d_\eta(A,A^m) - d_\eta(A,A^{det})} \\
&  \le d_\eta(A^m,A^{det}).\\
& = \frac{1}{2} \int \norm{M^m(\theta) - M^{det}(\theta)}  d\eta(\theta).
\end{align*}
Fix $\varepsilon > 0$, and define
\begin{equation}
S_h(A) = \{ \theta : \abs{A^{m}(p_h^\theta) - A^{det}(p_h^\theta)} > \varepsilon\}.
\end{equation}
In what follows, we mostly omit the dependence of $S_h$ on the admissions function $A$, except where necessary to avoid confusion. We claim that 
\begin{equation} \theta \not \in \bigcup_{h \in \H} S_h \Rightarrow \frac{1}{2} \norm{M^m(\theta) - M^{det}(\theta)} \le \varepsilon |\H|. \label{eq:eps-bound} \end{equation}
To see this, note that $M^{det}(\theta) \in \{0,1\}^{\H_0}$ has exactly one nonzero component. If $M^{det}_\emptyset(\theta) = 1$, then \eqref{eq:enroll-prob} implies that 
\begin{equation} \frac{1}{2}\norm{M^m(\theta) - M^{det}(\theta)} \le \sum_{h \in \H} A_h^m(p_h^\theta) \le \varepsilon |\H|,\label{eq:bound1} \end{equation}
with the second inequality following because $M_\emptyset^{det}(\theta) = 1$ implies $A_h^{det}(p_h^\theta) = 0$ for all $h \succ^\theta \emptyset$, and thus $A_h^{m}(p_h^\theta) \le \varepsilon$ by the assumption that $\theta \not \in S_h$.

Meanwhile, if $M_h^{det}(\theta) = 1$ for some $h \in \H$, then 
\begin{align}
\frac{1}{2}\norm{M^m(\theta) - M^{det}(\theta)} & = 1 - M_h^m(\theta)  \nonumber\\
& = 1 - A_h^m(p_h^\theta) \prod_{h' \succ^\theta h} (1 - A_{h'}^m(p_{h'}^\theta)).\nonumber \\
& \le 1 - (1 - \varepsilon)^{|\H|} \nonumber \\
& \le \varepsilon |\H|. \label{eq:bound2}
\end{align}
The second line uses \eqref{eq:enroll-prob}, and the third uses that $\theta \not \in S_h$ for any $h$. Jointly, \eqref{eq:bound1} and \eqref{eq:bound2} imply \eqref{eq:eps-bound}. From this, it follows that 
%\begin{equation}S = \left\{ \theta : \frac{1}{2} \norm{M^m(\theta) - M(\theta)} >  \varepsilon |\H| \right\}. \label{eq:S-def} \end{equation}
%Then 
\begin{align*}
d_\eta(A^m,A^{det}) & =  \frac{1}{2} \int_{\theta \in \bigcup_{h \in \H} S_h} \norm{M^m(\theta) - M^{det}(\theta)} d\eta(\theta) +  \frac{1}{2} \int_{\theta \not \in \bigcup_{h \in \H} S_h} \norm{M^m(\theta) - M^{det}(\theta)} d\eta(\theta) \\
& \le \frac{1}{2} \int_{\theta \in \bigcup_{h \in \H} S_h} \norm{M^m(\theta) - M^{det}(\theta)} d\eta(\theta) + \varepsilon |\H|\eta(\Theta) \\
& \le \sum_{h \in \H} \eta(S_h) + \varepsilon |\H| \eta(\Theta),
\end{align*}
where the first line uses \eqref{eq:eps-bound} and the second uses a union bound and the fact that 
\[\frac{1}{2}\norm{M^m(\theta) - M^{det}(\theta)} \le \frac{1}{2} (\norm{M^m(\theta)} + ||M^{det}(\theta)||_1) = 1.\] Thus, all that remains is to show that we can bound $\eta(S_h)$ by a term that does not depend on the initial admissions function $A$, and goes to zero as $m \rightarrow \infty$. 

We let $I^A = \mI(\M(A))$, and prove this result in three steps. The first step is given in Lemma \ref{lem:poisson-bounds}, which implies that if $\theta \in S_h$, then $I_h^A(p_h^\theta)$ is ``close" to $C_h$. More precisely,
\begin{equation} \abs{I_h^A(p_h^\theta) - C_h} \le \sqrt{2 C_h \log(1/\varepsilon)/m} + 4\log(1/\varepsilon)/m.\label{eq:interest-bound} \end{equation}

The second step is to conclude that $p_h^\theta$ must fall in a ``small" interval. That is, if we let $[\underline{p}_h^A, \overline{p}_h^A]$ be the interval of priority scores $p_h^\theta$ satisfying \eqref{eq:interest-bound}, we wish to show that $\overline{p_h}^A - \underline{p}_h^A$ can be bounded by a term that does not depend on $A$ and goes to zero as $m$ grows. That is, as $m \rightarrow \infty$,
\begin{equation} \sup_{A \in \Admissions} (\overline{p}_h^A - \underline{p}_h^A) \rightarrow 0. \label{eq:small-interval} \end{equation}

Our third step is to note that because $\eta$ has strict priorities, it immediately follows from \eqref{eq:small-interval} and the fact that $S_h(A) \subseteq \{\theta : p_h^\theta \in [\underline{p}_h^A, \overline{p}_h^A]\}$ that $\sup_{A \in \Admissions} \eta(S_h(A)) \rightarrow 0$ as $m \rightarrow \infty$.

All that remains is to establish \eqref{eq:small-interval}. Let $I^L$ be the interest function that results when starting with the admissions function $A_h(p) = 1$ for all $h, p$. The superscript $L$ is a reminder that because the functions $\M$ and $\mI$ are monotone according to the orders defined in Section \sref{sec:existence}, the interest function $I^L$ is pointwise smaller than the interest function that results from any other admissions function $A$. In fact, the definition of $\mI$ in \eqref{eq:continuum-interest} implies the stronger claim that for any $h \in \H$ and $0 \le \underline{p} \le \overline{p} \le 1$,
\begin{equation}I_h^L(\underline{p}) - I_h^L(\overline{p}) \le I_h^A(\underline{p}) - I_h^A(\overline{p}).  \end{equation} 
In particular, this holds for $\underline{p} = \underline{p}^A$ and $\overline{p} = \overline{p}^A$. By the definition of $\underline{p}^A$ and $\overline{p}^A$, it follows that
\[ I_h^L(\underline{p}^A) - I_h^L(\overline{p}^A) \le 2(\sqrt{2 C_h \log(1/\varepsilon)/m} + 4\log(1/\varepsilon)/m). \]
To conclude from this that $\overline{p}^A - \underline{p}^A$ must be small for all $A \in \Admissions$, we invoke the assumption that $\eta$ has full support, which implies that for each $h \in \H$, $I_h^L$ is {\em strictly} decreasing (that is, there are students of every priority level who list $h$ first).\footnote{If $\eta$ does not have full support, then $\Delta^m$ may not converge to $\Delta$. To see this, consider the following example. In $\mE$ there are two schools, with $C_1 = C_2 = 1$. All students prefer 1 to 2, and $\eta(\Theta) = \mu > 2$. Priorities at the two schools are uniformly distributed on the line $\{\theta : p_1^\theta + p_2^\theta = 1\}$. The unique $(\mE,\vdet)$-stable outcome is a cutoff of $1 - 1/\mu$ at both schools (both schools get their favorite students).

Define $A$ by $A_1(p) = {\bf 1}(p \ge 1/\mu)$ and $A_2(p) = 1$. Then the resulting interest functions are $I_1^A(p) = \mu(1-p)$, and $I_2^A(p) = \min(\mu(1-p),1)$. In particular, $I_2^A(p) = 1$ for $p \in [0,1 - 1/\mu]$. 

Thus $A^{det}(I)$ rejects students with priority below $1 - 1/\mu$ at both schools (recall that to be accepted, we need $I_h(p) < C_h$). Note that this is actually the $(\mE,\vdet)$-stable outcome.

Meanwhile $A^{m}$ satisfies $A_1^m(p) = \mathbb{P}(Pois(m\mu(1-p))<m)$, and $A_2^m(p) \ge A_2^m(0) = \mathbb{P}(Pois(m)<m) \approx 1/2$. This implies that students with $p_1^\theta < 1 - 1/\mu$ are unlikely to be admitted to school $1$, but have approximately a 50\% chance (or higher) of being admitted to school 2 under $A^m$. Thus, the assignments under $A^{det}$ and $A^m$ differ for approximately $(\mu-2)/2$ students (more formally, $d_\eta(A^m, A^{det}) \rightarrow (\mu-2)/2$ as $m \rightarrow \infty$).
}
\end{proof}

\begin{lemma}
Fix $m, C \in \mathbb{N}$ and $\varepsilon > 0$. If $\lambda \in \mathbb{R}_+$ is such that $\abs{\vpois(m\lambda, mC)-\vdet(\lambda,C)} > \varepsilon$, then
\[ C - \sqrt{2 C \log(1/\varepsilon)/m} \le  \lambda \le C + \sqrt{2 C \log(1/\varepsilon)/m} + 4\log(1/\varepsilon)/m.\] \label{lem:poisson-bounds}
\end{lemma}
\begin{proof}[Proof of Lemma \tref{lem:poisson-bounds}]
In what follows we make use of a concentration bound for Poisson random variables\footnote{\url{http://www.cs.columbia.edu/~ccanonne/files/misc/2017-poissonconcentration.pdf}} which states that for any $\gamma, x > 0$,
\begin{align}
\mathbb{P}(Pois(\gamma) \ge \gamma + x) & \le exp\left(- \frac{x^2}{2(\gamma+x)}\right) \label{eq:poisson-upper} \\
\mathbb{P}(Pois(\gamma) \le \gamma - x) & \le exp\left(- \frac{x^2}{2(\gamma+x)}\right) \label{eq:poisson-lower}
\end{align}

First consider the case $\vdet(\lambda,C) = 1$, meaning that $\lambda < C$. In that case, 
\[\abs{\vpois(m\lambda, mC)-\vdet(\lambda,C)} > \varepsilon \Rightarrow \vpois(m\lambda,mC) < 1 - \varepsilon \Rightarrow \mathbb{P}(Pois(m\lambda) \ge mC) > \varepsilon.\] 
Applying \eqref{eq:poisson-upper} with $\gamma = m \lambda$ and $x = mC - m\lambda > 0$ and taking the log of both sides yields
\[ \log(1/\varepsilon) > \frac{x^2}{2(\gamma+x)} = \frac{m (C-\lambda)^2}{2C}.\]
Because $\lambda < C$, rearranging this inequality implies 
\[ \lambda \ge C -  \sqrt{2C \log(1/\varepsilon)/m}.\]

Next consider the case $\vdet(\lambda,C) = 0$, meaning that $\lambda \ge C$. In that case, apply \eqref{eq:poisson-lower} with $\gamma = m \lambda$ and $x = m\lambda - mC> 0$ and take logarithms to get that 
\[\vpois(m\lambda, mC) > \varepsilon \Rightarrow \log(1/\varepsilon) \ge \frac{x^2}{2(\gamma + x)} = \frac{m(\lambda-C)^2}{4\lambda - 2C}. \]
Because the denominator is positive, we can rearrange this inequality as
\[ (\lambda- C)^2 - 4\lambda \log(1/\varepsilon)/m + 2C \log(1/\varepsilon)/m \le 0.\]
This can be rewritten as
\[ (\lambda - C -  2\log(1/\varepsilon)/m)^2 \le 2C \log(1/\varepsilon)/m+(2 \log(1/\varepsilon)/m)^2,\]
from which it follows (using $\sqrt{a+b} \le \sqrt{a} + \sqrt{b}$ for positive $a, b$) that 
\[ \lambda \le C+ \sqrt{2C \log(1/\varepsilon)/m}+ 4\log(1/\varepsilon)/m.\]
\end{proof}

\begin{proof}[Proof of Proposition \tref{prop:convergence}]

The definition of $\Delta$ and the definition of stability in Definition \ref{def:stability} implies that if $A$ is $(\mE,\vdet)$-stable, then $\Delta(A) = 0$. Because $d_\eta$ is a metric by Lemma \ref{lem:metric}, the converse is also true: if $\Delta(A) = 0$, then $A$ is $(\mE,\vdet)$-stable. Similarly, the definition of $\Delta^m$ implies that $\Delta^m(A) = 0$ if and only if $A$ is $(\mE^m,\vpois)$-stable. 

Theorem 1 from \cite{azevedo-leshno_2016} implies that there is a unique $(\mE,\vdet)$-stable admissions function (root of $\Delta$), which we denote $A^*$. Theorem \tref{thm:uniqueness} implies that there is a unique $(\mE^m,\vdet)$-stable admissions function (root of $\Delta^m$), which we denote $A^m$. Lemma \ref{lem:metric} states that $d_\eta$ is a metric on $\mA$, and that $\mA$ is compact. Lemma \ref{lem:uniform-convergence-holds} states that $\Delta^m$ converges to $\Delta$ uniformly. Therefore, we may apply Lemma \tref{lem:uniform-convergence-helpful} to conclude that $d_\eta(A^m,A^*) \rightarrow 0$.
\end{proof}
}

\todo{
High level idea: $\vpois(m\lambda,mC_h) \rightarrow \vdet(\lambda,C_h)$ for $\lambda \neq C_h$.

Idea of proof: show that convergence is uniform. Fixing large $m$, for any $A$, predictions from the two models aren't very far from each other.

To show uniform convergence, want to show that measure of students whose predictions are ``significantly different" from two models goes to zero.
\begin{enumerate}
\item Show that if $m$ is large, then disagreement implies $I_h(p_h^\theta)$ is very close to $C_h$
\item Show that $I_h(p_h^\theta)$ close to $C_h$ implies that $p_h^\theta$ in a ``small" window of priority scores (uses full support).
\item Show that looking at a small window of priority scores implies small $\eta$-measure (uses strict priorities).
\end{enumerate}
}

%\na{
%Let $A^H$ be the admissions function that is identically one. Let $I^L$ to be $\mI(\M(A))$: that is, $I_h^L(p)$ gives the $\eta$-measure of students who list $h$ first and have priority $p_h^\theta > p$. 
%
%Note that by the assumption that $\eta$ has full support, $I_h^L$ is {\em strictly} decreasing on $[0,1]$ for each $h$. Therefore, if we define $f_\eta : \mathbb{R}_+ \rightarrow \mathbb{R}_+$ by 
%\[ f_\eta(\delta) = \sup \{ \gamma : \exists p \in [0,1], h \in \H \,\, s.t. \,\, I_h^L(p) - I_h^L(p+\gamma) \le \delta\},\]
%then $f_\eta(\delta) \rightarrow 0$ as $\delta \rightarrow 0$.
%
%Claim that for any $A$, $I^A_h(p) - I^A_h(p+\gamma) \ge I_h^L(p) - I_h^L(p+\gamma)$. Thus, $I^A_h(p) - I^A_h(p+\gamma) \le x$ implies $\gamma \le f_\eta(x)$ (apply for $x = \mathcal{O}(1/\sqrt{m})$). 
%
%Define 
%\[g_\eta(\gamma) = \sup_{p \in [0,1], h \in \H} \eta(\{ \theta : p_h^\theta \in [p,p+\gamma)\}).\]
%Because $\eta$ has strict priorities, we have $g_\eta(\gamma) \rightarrow 0$ as $\gamma \rightarrow 0$.
%
%It follows that $\eta(S_h) \le g_\eta(f_\eta(c_\eta/\sqrt{m})) \rightarrow 0$. (If we wanted could have $g_h^\eta$ and $f_h^\eta$, remove some quantifiers, but introduce some unnecessary notation).
%}

\na{Possibility that I probably don't want to pursue:
For each school $h$ in the original market, create $m$ copies in the new market (with identical capacities). Generate lists as in original market, except for each school on list, select one school uniformly at random among the copies of that school.
}

In the limit as the number of students and the capacity of each school grow, the predictions of this new model converge to those of \citet{azevedo-leshno_2016}.\footnote{To be more precise, given a market $\mathcal{E} = (\H, {\bf C}, \eta)$, we can define a sequence of markets $\mathcal{E}^m = (\H, {\bf C}^m, \eta^m)$, with $C_h^m= m C_h$ for all $h \in \H$ and $\eta^m(\tilde{\Theta}) = m \cdot \eta(\tilde{\Theta})$ for all measurable $\tilde{\Theta} \subseteq \Theta$. That is, we simply scale up the number of students and number of seats per school by a factor of $m$. This is the same limiting regime as considered by \citet{azevedo-leshno_2016}. If $\mathcal{E}$ has a unique $(\eta, \vdet)$-stable matching $\MStable$, then the $(\eta,\vpois)$-stable matchings of $\mathcal{E}^m$ converge to $\MStable$ as $m$ grows. This is because Poisson random variables with large means are highly concentrated. More formally, for any $C \in \mathbb{N}$ and any $\lambda \in \mathbb{R}_+ \backslash \{ C\}$, $\vpois(m\lambda,mC) \rightarrow \vdet(\lambda,C)\text{ as } m \rightarrow \infty$.
%\[ \vpois(m\lambda,mC) \rightarrow \vdet(\lambda,C)\text{ as } m \rightarrow \infty.\]
} Conceptually, this is because Poisson random variables with large means are highly concentrated. The advantage of our model lies in its ability to generate probabilistic predictions for markets where capacities are modest and cutoffs are uncertain. We demonstrate its accuracy by using it to generate predictions about two measures of interest: the measure of matched students
\begin{equation} \sum_{h \in \H} \int  M_h(\theta) d\eta(\theta), \end{equation}
 and the average rank of matched students, which we now define.  Denote $\theta$'s rank of $h \in \H$ by
\[R_h(\theta)= |\{h' \in \H_0 : h' \succeq^\theta h\}|\]
and define\footnote{Note that \eqref{eq:avg-rank} counts the average rank among matched students. We make this choice to align with the convention adopted by \citet{ashlagi-kanoria-leshno_2017}. One could instead work with the average rank across {\em all} students, counting a student who lists $\ell$ schools and goes unassigned as receiving their $(\ell+1)^{st}$ choice. In general, these two definitions are incomparable, meaning that each can be larger than the other. However, Propositions \tref{prop:more-seats} holds for either definition. To see this, note that 
\[ \frac{\int \sum_{h \in \H_0} M_h(\theta) R_h(\theta) d\eta(\theta)}{\eta(\Theta)} = \frac{\sum_{h \in \H} \I_h(0)+\I_\emptyset(0)}{\sum_{h \in \H} \int  M_h(\theta) d\eta(\theta) + \I_\emptyset(0)}  \leq \frac{\sum_{h \in \H} \I_h(0)}{\sum_{h \in \H} \int  M_h(\theta) d\eta(\theta)}. \]
}
%and let $\mathcal{R}^\theta(M)$ be the expected rank of type $\theta$, given admissions probabilities $A$: \na{$\R$ was already used above (to match AL notation).}
%\begin{equation} \mathcal{R}^\theta(M) = \sum_{h \in \H_0} M_h(\theta) R_h^\theta. \label{eq:Rtheta} \end{equation}
\begin{equation} \averagerank(M) = \frac{\int \sum_{h \in \H} M_h(\theta) R_h(\theta) d\eta(\theta)}{\int \sum_{h \in \H} M_h(\theta)  d\eta(\theta)}. \label{eq:avg-rank} \end{equation}
% \frac{\int \sum_{h \in \H_0} M_h(\theta) R_h^\theta d\eta(\theta)}{\eta(\Theta)}.

%\begin{equation} \first_k(\I) = n \E_{\theta \sim \eta}\left[ \sum_{j = 1}^{\min(k,\ell(\theta))} \mI_{h_j^\theta}^\theta(\mA(\I)) \mA_{h_j^\theta}^\theta(\I) \right]. \label{eq:fk}\end{equation}
%\begin{equation} \averagerank(\I) = \frac{\int_\Theta \E_{\Q(\I)}[R^\theta {\bf 1}(R^\theta \leq \ell(\theta))] d\eta }{\int_\Theta \p_{\Q(\I)}(R^\theta \leq \ell(\theta)) d\eta}.\label{eq:avg-rank} \end{equation}
%\na{quantity $\first_k(\I)$ gives the expected number of students who are admitted to one of the first $k$ schools on their list, and the }

We compare our predictions to existing analytical and simulation results for finite markets. To predict outcomes from random finite markets with $n$ students whose types are drawn drawn iid from a probability distribution $\tilde{\eta}$ on $\Theta$, we define the measure $\eta$ by $\eta(\tilde{\Theta}) = n \tilde{\eta}(\tilde{\Theta})$ for all $\tilde{\Theta} \subseteq \Theta$, and study $(\eta,\vpois)$-stable matchings.

%Prior work almost exclusively studies the case where student preferences are drawn uniformly at random. In our model, this corresponds to having one type of school. To facilitate the comparison, the analysis that follows considers this special case.

Section \sref{sec:avg-rank} provides results for students' average rank, using the results of \citet{ashlagi-kanoria-leshno_2017} as the primary comparison. Section \sref{sec:matched-students} provides results for the number of matches, and compares against findings from \citet{marx-schummer_2021}. In both cases, our model accurately predicts simulation results for finite markets of moderate size. In addition, our model provides new analytical expressions and insights, described below. %Regarding average rank, we provide bounds that match those from \citet{ashlagi-kanoria-leshno_2017}, but apply for {\em any} priority rule. For match size, we note that \citet{arnosti_2021} uses our model to analytically compare match size under different priority rules, and provides clean expressions for these quantities. 

\na{Possibility that I probably don't want to pursue:
For each school $h$ in the original market, create $m$ copies in the new market (with identical capacities). Generate lists as in original market, except for each school on list, select one school uniformly at random among the copies of that school.
}
%}

%\section{Model Predictions} \label{sec:insights}

%Several previous papers have studied outcomes in finite markets when preferences are iid uniform. In this section, we explain how to use our model to generate numerical and analytical predictions about this setting. We focus especially on the work of \citet{ashlagi-kanoria-leshno_2017}, which studies the average rank of matched students. In our model, this is defined as follows. 

\subsection{Average Rank} \label{sec:avg-rank}

We now present our results on average rank, using as a comparison the simulations and numerical bounds from \citet{ashlagi-kanoria-leshno_2017}. We first introduce some new definitions. 
%\citet{ashlagi-kanoria-leshno_2017} consider a setting where each school has a single seat and both sides have complete preferences drawn uniformly at random. Their simulation results for markets with $40$ schools (``women" in their paper) and anywhere from 20 to 60 students (``men"). 

%They conduct an asymptotic analysis as the number of participants on both sides grows, and establish that the average rank of students is low when there are fewer students than schools, and jumps considerably when there is even a single extra student. To demonstrate that these findings hold in markets of reasonable size, they conduct simulations for markets with $40$ schools (``women" in their paper) and anywhere from 20 to 60 students (``men"). 

\begin{definition} \text{ }
%A market $(\H, {\bf C}, \eta)$ is {\bf symmetric} if (i) $C_h = C_{h'}$ for all $h,h' \in \H$, and (ii) the restriction of $\succ^\theta$ to $\H$ is uniformly distributed. 

The measure $\eta$ is a {\bf symmetric IID measure} if (i) the restriction of $\succ^\theta$ to $\H$ is uniformly distributed, and (ii) for each $\succ \in \mathcal{R}$, the conditional distribution of ${\bf p}^\theta$ given $\succ^\theta = \succ$ is uniform on $[0,1]^\H$. 

The measure $\eta$ is a {\bf symmetric RSD measure} if (i) the restriction of $\succ^\theta$ to $\H$ is uniformly distributed, and (ii) for each $\succ^\theta \in \mathcal{R}$,  the conditional distribution of ${\bf p}^\theta$ given $\succ^\theta = \succ$ is uniform on $\{\theta : p_h^\theta = p_{h'}^\theta \text{ for all } h, h' \in \H\}$.

%A market $(\H, {\bf C}, \eta)$  has {\bf RSD priorities} if for each $\succ^\theta \in \mathcal{R}$,  the conditional distribution of ${\bf p}^\theta$ given $\succ^\theta = \succ$ is uniform on $\{\theta : p_h^\theta = p_{h'}^\theta \text{ for all } h, h' \in \H\}$.
\end{definition}

%market with $|\H| = 40$ schools, each with a single seat. We take $\eta$ to be a uniform measure over the subset of $\mathcal{R} \times [0,1]^\H$ such that all schools are acceptable, with $\eta(\Theta)$ ranging from $20$ to $60$. Finally, we let 

\begin{figure}
\captionsetup{width=\textwidth}
\centerline{\includegraphics[width = 0.8 \textwidth]{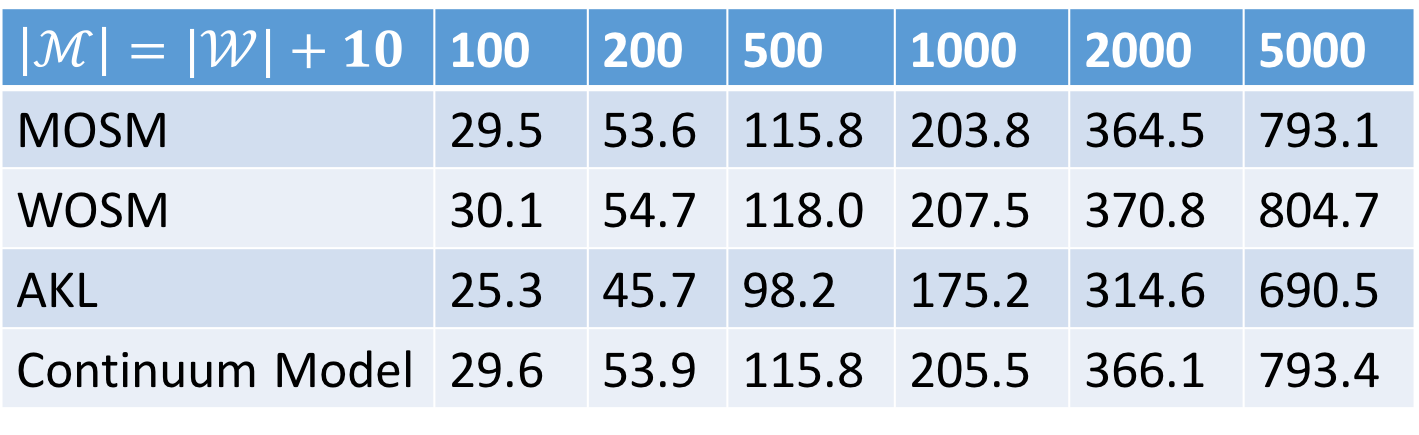}}
\caption{A reproduction of the first column of Table 1 from \citet{ashlagi-kanoria-leshno_2017}, along with our model predictions. Each column above corresponds to a different market size, holding fixed the absolute imbalance. The first two rows show the (simulated) average rank under the student optimal (``man optimal" in their paper) and school optimal (``woman optimal") stable matches, which differ by at most 2\%. The third row shows predictions made by \citet{ashlagi-kanoria-leshno_2017}, which underestimate the average rank by approximately 15\%. The final row shows our predictions, which always lie between the simulation results for the extremal stable matchings. \todo{Relabel MOSM/WOSM?} } \label{fig:akl-table1}
\end{figure} 

%We wish to use our model to make predictions about these simulations. To do so, we use the following parameter choices. There is a single school type with capacity $C = 1$, and school mass $\D = 40$. The total mass of students ranges from $\eta(\Theta) = 20$ to $\eta(\Theta) = 60$.\footnote{The predictions our model are scale invariant, meaning that simultaneously doubling the mass of students $\eta(\Theta)$ and of schools $\D$ does not change the predictions. Thus, instead of choosing $\D = 40$ and letting the mass of students $\eta(\Theta)$ range from $20$ to $60$, one could equivalently normalize $\D = 1$ and let $\eta(\Theta)$ range from $0.5$ to $1.5$.}  All students submit lists of length $\ell = 40$, and student priorities are drawn uniformly on $[0,1]^{\ell}$.

To predict the outcome of their simulations, we consider a market with $C_h = 1$ for all $h$, and $\eta$ the symmetric iid measure in which all schools are preferred to the outside option, and analyze $(\eta, \vpois)$-stable matchings. There are several reasons that our model's predictions might not match the simulation results. First, our model generates a unique prediction, whereas finite markets may have multiple stable matchings. Although \citet{ashlagi-kanoria-leshno_2017} establish that differences between stable matchings are relatively minor in imbalanced markets, our ``prediction error" must at least be comparable to the variation across stable matchings. Second, the assumption of independent outcomes across schools introduces error: when the number of students is below 40, every student in the finite market must match, whereas our model predicts that each student has a positive (albeit small) probability of going unassigned.

Despite these concerns, the model prediction is excellent. Figure \fref{fig:akl2} displays their reported simulation results alongside our predictions. The curves do not merely appear qualitatively similar, they also match quantitatively. To emphasize this point, we reproduce their Table 1 in Figure \fref{fig:akl-table1}. In their simulations, the average rank of students differs by at most 2\% between student-optimal and school-optimal stable matches. Our model predictions always lie in between these two numbers. 

Although the numerical predictions of our model are excellent, the resulting formulas are complex, and do not immediately offer insight about how students' average rank depends on market primitives. We address this by using our model to derive analytical bounds on students' average rank. Figure \fref{fig:akl2} makes it clear that the behavior of the market is very different depending on whether the number of students is less or greater than the number of seats. Accordingly, our analysis will consider these cases separately. We define $\rho$ to be the ratio of students to schools:
\begin{equation} \rho = \eta(\Theta)/|\H|, \end{equation}
%\na{Should we instead let $\rho = n/\sum_{h \in \H} C_h$?}
and let $C$ denote the (common) capacity at each school. We analyze the case with more seats ($\rho < C$) in Section \sref{sec:more-seats}, and the case with more students ($\rho > C$) in Section \sref{sec:more-students}.

\subsubsection{More Seats than Students} \label{sec:more-seats}

We define
\begin{equation} \enrollment(\lambda,C) = \int_0^\lambda \mV(x,C) dx. \label{eq:enrollment} \end{equation}

To state our bounds, for $C \in \mathbb{N}$ and $\rho < C$, define $\Lambda(\rho,C)$ to be the smallest solution $\lambda$ to 
\[ \enrollment(\lambda,C) = \rho.\]
Note that the definition of $\vpois$ in \eqref{eq:poisson} implies that for any $C \in \mathbb{N}$,
\begin{equation} \int_0^\infty \vpois(\lambda,C) = C,\label{eq:fill-to-capacity} \end{equation}
so $\Lambda(\rho,C)$ is defined for $\rho < C$.
%\begin{equation} \int_0^{\Lambda(\rho,C)} \mV(x,C) dx= \rho.\end{equation} 
 
\begin{proposition} \label{prop:more-seats}
Fix $C \in \mathbb{N}$, and let $C_h = C$ for all $h \in \H$. Let $\eta^{IID}$ be a symmetric iid measure, and let $M^{IID}$ be the unique $(\eta^{IID},\vpois)$-stable matching guaranteed by Theorem \tref{thm:uniqueness}. If there are more seats than students ($\rho < C$), then 
\[\averagerank(M^{IID}) \leq \Lambda(\rho,C)/\rho.\]
\end{proposition}

To clarify the relationship with results from \citet{ashlagi-kanoria-leshno_2017}, we parameterize $\rho$ and apply Proposition \tref{prop:more-seats} to the special case where $C = 1$.
\begin{corollary} \label{cor:akl}
In a symmetric iid market where schools have a single seat and $\rho = \frac{n}{n+k}  < 1$, 
\[\averagerank(M^{IID}) \leq \frac{n+k}{n} \log \left( \frac{n+k}{k} \right).\]
%= \frac{1}{1+\lambda}
%\[\averagerank(\I) \leq \frac{n+k}{n} \log \left( \frac{n+k}{k} \right) = (1+\lambda)\log(1+1/\lambda).\]
\end{corollary}
This upper bound exactly matches that from Theorem 2 of \citet{ashlagi-kanoria-leshno_2017}.\footnote{Also matches heuristic bound derived in conclusion of \citet{wilson_1972} when $\rho < 1$.} Because we work with different models, neither result directly implies the other.  \todo{Revisit commented text below.}
\remove{However, Corollary \ref{cor:akl} is ``stronger" in that it makes no assumptions on how priorities are generated. \citet{ashlagi-kanoria-leshno_2017} note that the average rank under random serial dictatorship is similar to that with independent priorities. Our result applies to these cases, as well as (for example) to cases where students systematically have low priority at their first choice school, where a high priority at one school implies lower priority at the remaining listed schools, or where students with long lists have systematically better (or worse) priority scores than those with short lists.\todo{\footnote{Explain balls in bins idea. Also point out that as long as most students are matching, this implies a nearly matching lower bound. For example, if students list identical number of schools, then the average rank is a deterministic function of the match rate (independent of priorities).\todo{Mention \citet{knoblauch_1999}, who shows that all priorities result in average rank asymptotic to $\log(n)$.}}} \todo{Mention \citet{knuth_1976,knuth_1997}, who proved this result. Also perhaps proven earlier by \citet{wilson_1972}.}}
However, the bound in Proposition \tref{prop:more-seats} provides insight beyond the cases considered by \citet{ashlagi-kanoria-leshno_2017}. First, it does not assume that students submit complete (or long) lists: the distribution of list lengths can be arbitrary.\footnote{While it is well known that increasing a student's list makes outcomes worse for all other students, this does not imply that extending a list increases the average rank, because average rank is calculated only for matched students. If some students submit short lists, extending their list may cause these students to match in place of (or in addition to) others with longer lists, thereby decreasing the average rank.} Instead, fixing school capacity $C$, the bound depends only on the ratio of students to schools $\rho$. Second, our bound improves as  school capacity grows: fixing the ratio of students to seats $\rho/C = 0.97$, the bound on average rank is approximately $3.6$ when $C = 1$,  $2.0$ when $C = 3$, $1.4$ when $C = 10$, and tends to $1$ as $C \rightarrow \infty$.

%applies when schools have more than one seat, and clarifies the role of market primitives. In the model of \citet{ashlagi-kanoria-leshno}, the parameter $n+k$ represents both the number of schools and the length of student lists. Proposition \tref{prop:more-seats}

\subsubsection{More Students than Seats} \label{sec:more-students}

We now turn to the case where students outnumber seats. %In this case, the choice of priorities matters. 
We consider the two special cases studied by \citet{ashlagi-kanoria-leshno_2017}: priorities are either generated iid $U[0,1]$ (IID), or are identical across schools, and independent of the length of student lists (RSD). The following result implies that RSD results in a much lower average rank than IID. 
\begin{proposition} \label{prop:more-students}
%If $\rho < C$, the average rank of students is at most $\enrollment^{-1}(\rho)/\rho$. \\
%\vspace{-.2 in}
%\hspace{1 in} 
Suppose that $C_h = C_{h'}$ for all $h, h' \in \H$ and that there are more students than seats ($\rho > C$). 
\noindent Let $\eta^{IID}$ be a symmetric IID measure in which all students list $\ell$ schools, and let $M^{IID}$ be the unique $(\eta^{IID}, \vpois)$-stable matching. Then
\[ \averagerank(M^{IID}) \geq \ell \left(1 - \frac{\rho}{C}- \frac{1}{\log \left( 1-C/\rho \right)} \right) .\]
Let $\eta^{RSD}$ be a symmetric RSD measure in which all students list $\ell$ schools, and let $M^{RSD}$ be the unique $(\eta^{RSD}, \vpois)$-stable matching. Then
\[ \averagerank(M^{RSD}) \leq 1+\log(\ell). \hspace{1.3 in}\]
\end{proposition}
%Furthermore, one might expect that the average rank with IID priorities would improve with larger schools, but their analysis assumes $C = 1$.

To facilitate a comparison with bounds from \citet{ashlagi-kanoria-leshno_2017}, we parameterize $\rho$ and consider the special case where $C = 1$.
\begin{corollary} \label{cor:akl-lb}
In a symmetric iid market where $C = 1$, $\rho = \frac{n+k}{n}$ and all students list $n$ schools, 
\[\averagerank(M^{IID}) \geq \frac{n}{\log\left( \frac{n+k}{k} \right)} -k .\]
\end{corollary}

This lower bound is very similar to the bound of $\frac{n}{1+\frac{n+k}{n}\log\left( \frac{n+k}{k} \right)}$ from Theorem 2 of \citet{ashlagi-kanoria-leshno_2017} (in fact, our bound is tighter for $k \leq n$).\footnote{Letting $\beta = k/n$, algebra reveals that the bound in Corollary \ref{cor:akl-lb} is tighter than the bound from Theorem 2 of \citet{ashlagi-kanoria-leshno_2017} so long as $(1+\beta)\beta \log^2(1+\beta) \leq 1$, which holds for $\beta \leq 1$.} More importantly, Proposition \tref{prop:more-students} clarifies the effect of list length, school capacity, and market size.

In the model of \citet{ashlagi-kanoria-leshno_2017}, $n$ represents both the number of schools and the length of student lists, and it is assumed that each school has a single seat. It is not a priori clear how their bound changes if students list only a subset of the market, or if schools have multiple seats. Proposition \tref{prop:more-students} establishes that students' average rank under RSD is logarithmic in the {\em list length} (rather than the market size). Meanwhile, under IID priorities, students' average rank is linear in the list length, with proportionality constant that depends on the ratio of students to seats $\rho/C$. This implies that with IID priorities, large capacities don't result in meaningfully better outcomes: given a fixed ratio of students to seats $\rho/C$, the lower bound does not depend on whether schools are small or large.\footnote{This is in contrast to RSD. In this case, fixing the ratio of students to seats $\rho/C$, the average rank is decreasing in $C$, and converges to one as $C$ grows (although this fact is not reflected in the bound in Proposition \tref{prop:more-students}).}

\subsection{Number of Matches} \label{sec:matched-students}

Another metric of interest is the number of matches. Few theoretical papers study this quantity, despite its salience in many applications. In fact, many papers assume that students submit complete lists, or at leasts lists that are long enough that the short side of the market matches fully. 

One recent exception is \citet{marx-schummer_2021}. They consider the problem facing a matching platform that helps to match men and women, and charges prices $p_m$ and $p_w$ to each matched man and woman, respectively. A man and woman can only be matched if they are both willing to pay the fee, which occurs with probability $1 - q(p_m,p_w)$ independently across man-woman pairs. The platform thus faces a tradeoff: if its prices are too high, there will be few mutually acceptable pairs, and few matches will form. The goal of the platform is to choose prices to maximize its revenue.

%The main technical challenge that they confront is to study the number of matches that form at a given price. 
\citet{marx-schummer_2021} assume that both sides have uniformly random preferences over mutually acceptable partners, and that the final matching is stable. Let $V^{IID}(W,M,q)$ be the expected number of matches with $W$ women, $M$ men, and probability of mutual compatibility $1 - q$. The main technical challenge they confront is analyzing $V^{IID}(W,M,q)$. They argue that an analytical result is intractible, and instead analyze an algorithm in which men declare all acceptable partners, and then women are placed in a random order and sequentially allowed to choose their favorite unmatched mutually acceptable man. In essence, they study the outcome of a woman-selecting random serial dictatorship. Letting $V^{RSD}(W,M,q)$ be the expected number of matches that form in this case, \citet{marx-schummer_2021} note that
\begin{equation} V^{RSD}(W,M,q) = \sum_{j = 1}^{\min(M,W)} \frac{\prod_{i = 0}^{j-1} (1 - q^{M-i}) \prod_{i = 0}^{j-1} (1 - q^{W-i}) }{1-q^j}. \label{eq:complex} \end{equation}
They use simulations to argue that $V^{RSD}(W,M,q)$ is a reasonable approximation of $V^{IID}(W,M,q)$. Figure \fref{fig:ms12} includes two plots from their paper: the first displays $V^{RSD}$, and the second displays simulation results showing the relative difference between $V^{RSD}$ and $V^{IID}$.

To generate corresponding predictions from our model, we let $W$ be the mass of students, $M$ be the number of schools, and assume that schools are symmetric and have capacity $C = 1$. In their simulations, student list length is Binomial with parameters $M$ and $1-q$. When $q$ is large, this is close to a Poisson with mean $M(1 - q)$, so we generate predictions assuming that student list lengths follow the latter distribution. We let $\hat{V}^{RSD}(W,M,q)$ and $\hat{V}^{IID}(W,M,q)$ be the mass of matches in the unique $(\eta^{RSD}, \vpois)$ and $(\eta^{IID}, \vpois)$ stable matchings of this market, and plot our predictions in Figure \fref{fig:ms12}, alongside the original graphs from \citet{marx-schummer_2021}. Despite the error in approximating the binomial distribution with a Poisson, the model predictions are very close to the exact expression and the simulation results. \todo{Could also generate predictions for binomial distribution -- probably not worth it.}

\begin{figure}
\captionsetup{width=\textwidth}
\centerline{\includegraphics[height = 1.6 in]{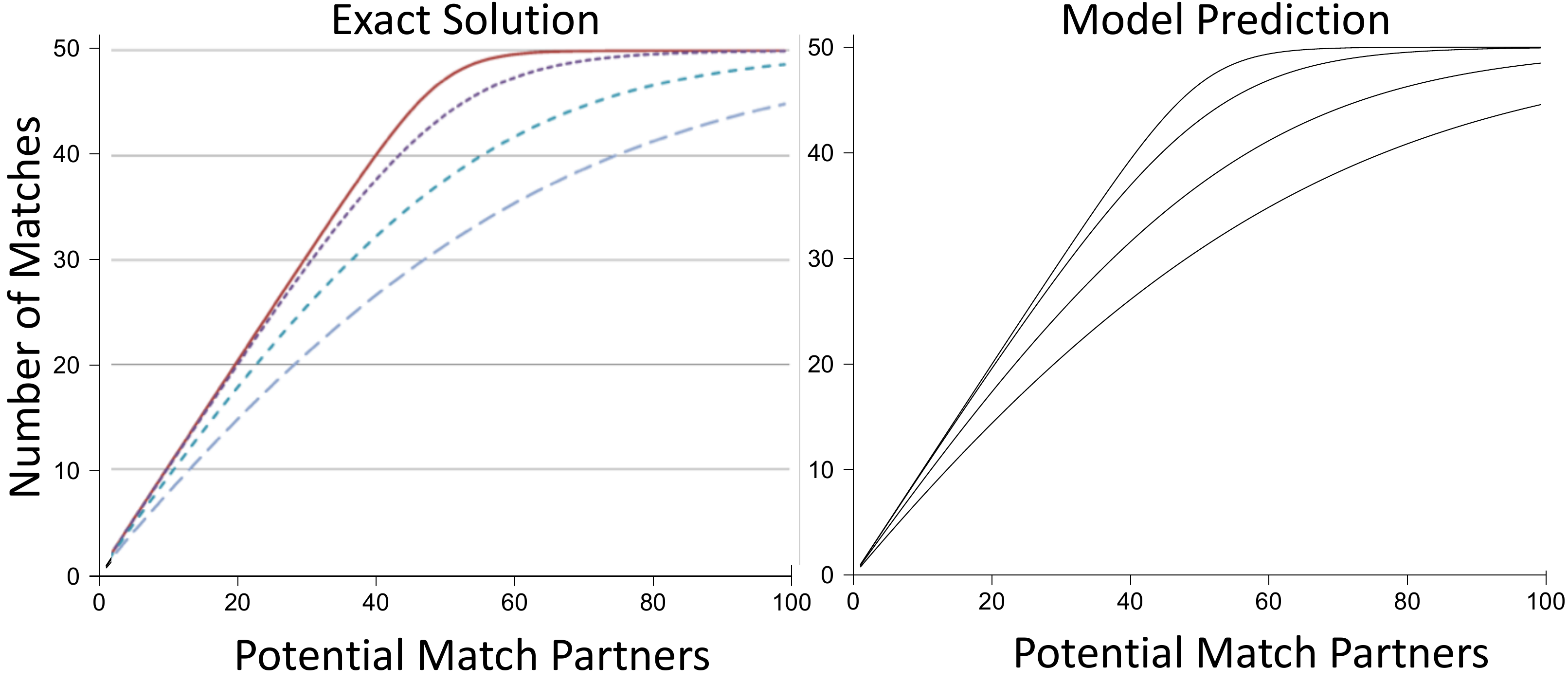}} \vspace{.1 in}
\centerline{\includegraphics[height = 1.6 in]{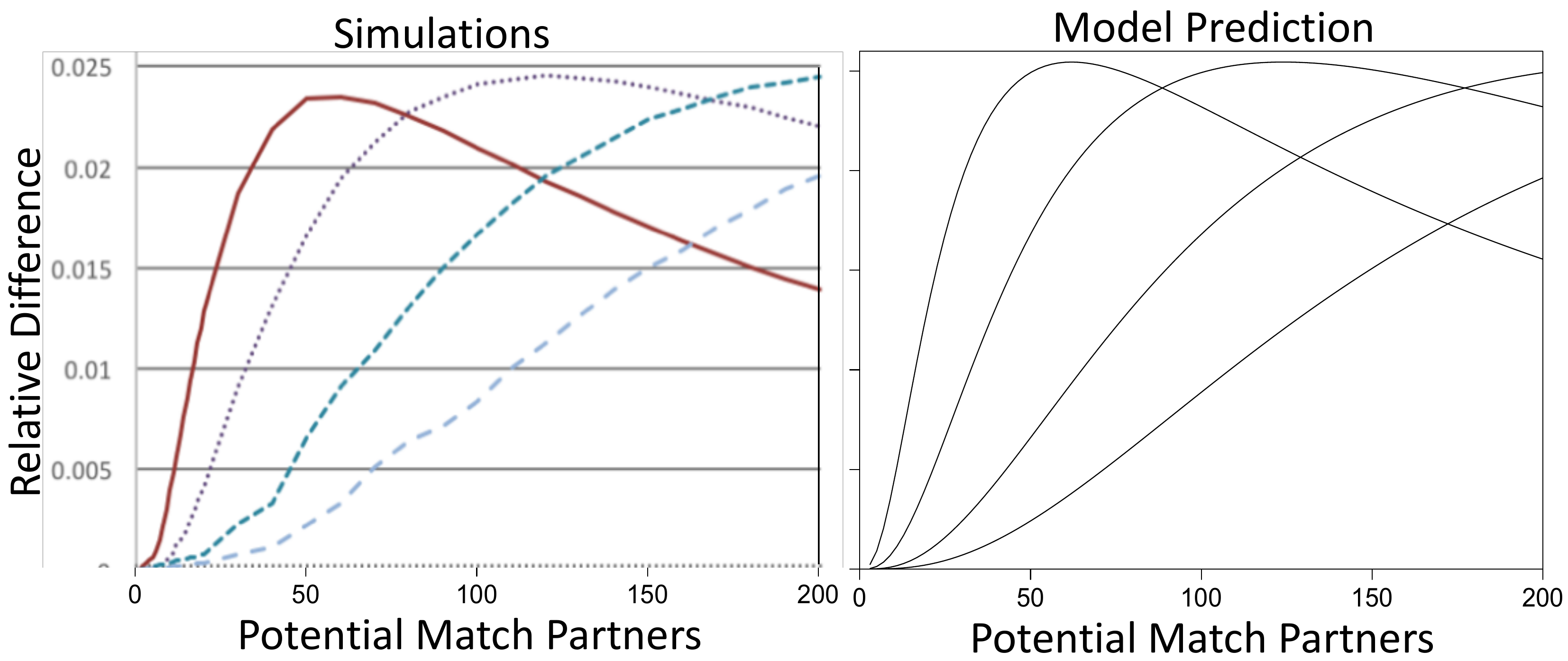} \vspace{-.1 in}}
%\centerline{\includegraphics[width = .48 \textwidth]{M-S-Fig1} \includegraphics[width = .48 \textwidth]{M-S-Fig1-dup}}
%\centerline{\includegraphics[width = .5 \textwidth]{M-S-Fig2}\includegraphics[width = .5 \textwidth]{M-S-Fig2-dup}}
\caption{Studying the number of matches that form, with 50 participants on one side and varying the number of participants on the other side. Upper left: exact expressions from \citet{marx-schummer_2021} for $V^{RSD}$ (each line represents a different probability of incompatibility $q$). Upper right: the corresponding predictions $\hat{V}^{RSD}$ from our proposed continuum model. Lower left: \citet{marx-schummer_2021} have no analytical expression for $V^{IID}$, but their simulations study the relative difference $(V^{IID} -V^{RSD} )/V^{RSD}$. In our proposed continuum model, it is possible to prove that $\hat{V}^{IID} \geq \hat{V}^{RSD}$: the lower right plot shows the predicted relative difference. By comparison, the continuum model of \citet{azevedo-leshno_2016} predicts that $V^{IID}=V^{RSD}$.
%Left: Figures 1 and 2 from \citet{marx-schummer_2021}. Top left shows expected match size $V^{RSD}(W,50,q)$ as a function of $W$, for varying $q$. \hspace{.7 in} Bottom left shows simulation-based estimates of $1 - V^{RSD}(W,50,q)/V^{IID}(W,50,q)$. \hspace{.5 in}Right: model predictions $\hat{V}^{RSD}(W,M,q)$ \,\, and \,\, $1 - \hat{V}^{RSD}(W,M,q)/\hat{V}^{IID}(W,M,q)$. Whereas the expression for $V^{RSD}$ is very complex and there is no known expression for $V^{IID}$, \eqref{eq:simpler} and \eqref{eq:iid} provide simple analytical expressions for $\hat{V}^{RSD}$ and $\hat{V}^{IID}$.
} \label{fig:ms12} 
\end{figure}

Figure \fref{fig:ms12} suggests that IID preferences result in more matches than RSD. This can be proven using our model. \citet{arnosti_2021} derives expressions which correspond to the predictions of our model with vacancy function $\vpois$, when student preferences are uniformly random and school priorities are either IID or RSD. Theorem 3 and Proposition 1 from \citet{arnosti_2021} imply that if the list length distribution has an increasing hazard rate, then IID priorities result in more matches than RSD. Both the binomial and Poisson distributions have an increasing hazard rate.

In addition to being numerically accurate, the predictions $\hat{V}^{RSD}(W,M,q)$ and $\hat{V}^{IID}(W,M,q)$ are analytically tractable. \citet{arnosti_2021} shows that when $W \leq M$,\footnote{The choice $W\leq M$ is without loss of generality, and cleans up the expression in \eqref{eq:simpler} by allowing the use of $W$ in place of $\min(W,M)$ and $M$ in place of $\max(W,M)$.}
\begin{equation} \hat{V}^{RSD}(W,M,q) = W - \frac{\log(1 + e^{-(M-W)(1-q)}- e^{-M(1-q)})}{1-q}.\label{eq:simpler} \end{equation}
This expression is much simpler than that in \eqref{eq:complex}, and more amenable to optimization. Furthermore, whereas calculating $V^{IID}(W,M,q)$ is intractable, \citet{arnosti_2021} shows that $\hat{V}^{IID}(W,M,q)$ is the unique solution to 
\begin{equation} (1 -q)\hat{V}^{IID}(W,M,q) = \log\left(1 - \frac{\hat{V}^{IID}(W,M,q)}{W} \right)  \log\left(1 - \frac{\hat{V}^{IID}(W,M,q)}{M} \right). \label{eq:iid} \end{equation} 
Note that for any desired match rate $V \leq \min(W,M)$, \eqref{eq:iid} gives a concise closed-form expression for the corresponding probability of incompatibility $q$.

\newpage
\section{Conclusion}
Stable matching algorithms are used to assign students to schools in cities across the globe. In theory, the design of school priorities offers a flexible tool for encoding policy objectives.  In practice, the benefits of designing priorities are limited by the fact that the relationship between priorities and the final outcome is complex and poorly understood. 

This paper provides a model that can begin to address these questions. Our model has three desirable features: it is {\em flexible} enough to accommodate complex preferences and priorities, its numerical predictions are extremely {\em accurate}, and it tractable enough to offer new {\em insights}. We use a novel framework for stable matching to show that the only difference between our model and that of \citet{azevedo-leshno_2016} is that they assume that interest at each school is deterministic, whereas we assume that it follows a Poisson distribution. This difference allows our model to make probabilistic predictions that reflect the uncertainty in finite random markets.

%As a result, our model makes similar predictions to theirs when school capacities are large, but produces much more accurate predictions with small or moderate capacities.

Much work remains, including the establishment of rigorous accuracy guarantees in settings with small school capacities. However, the formal guarantees of our model match those of \citet{azevedo-leshno_2016}, and its numerical accuracy is far superior when school capacities are modest. 
We believe that this model offers a fundamentally new perspective on stable matching, which will enable the study of settings -- such as those with small and asymmetrical schools -- which cannot be readily studied using prior techniques.

\bibliographystyle{ACM-Reference-Format}
\bibliography{Bibliography}

  \appendix

\renewcommand\thesection{\Alph{section}}
\setcounter{section}{0}

\section{Proofs from Section \ref{sec:results}}

\subsection{Proof of Proposition \tref{prop:finite}} \label{app:prop-finite}

The following Lemma states that for any $(\eta^\S, \vdet)$-stable outcome, the enrollment at each school will be the minimum of the school's capacity and the number of interested students. %relates the number of students matched to school $h$ to the amount of interest in $h$. It says

\begin{lemma} \label{lem:enrollment-finite}
If $\S$ is a finite subset of $\Theta$ with strict priorities, and $(M,I,A)$ is $(\eta^\S,\vdet)$-stable, then for any $h \in \H$ and $p \in [0,1]$, 
\[\sum_{\theta \in \S} {\bf 1}(p_h^\theta > p) M_h(\theta) = \min (I_h(p),C_h). \]
\end{lemma}

\begin{proof}[Proof of Lemma \ref{lem:enrollment-finite}]
%For any $\theta$ and any $h$, \eqref{eq:enroll-prob} implies that 
%\begin{equation} M_h(\theta) = A_h(p_h^\theta) \prod_{h' \succ^\theta h} (1 - A_{h'}(p_{h'}^\theta)) = A_h(p_h^\theta) (1- \sum_{h' \succ^\theta h} M_{h'}(\theta) ) \label{eq:unravel} \end{equation}
Define $\bar{p} = \inf \{p \in [0,1] : I_h(p) < C_h\}$. Note that \eqref{eq:continuum-interest} and \eqref{eq:etaS} imply that $I_h$ is right-continuous, and therefore $I_h(\bar{p}) < C_h$. It follows from \eqref{eq:admissions} and \eqref{eq:vdet} that for $p \in [0,1]$,
\begin{equation}A_h(p) = \vdet(I_h(p),C_h) = {\bf 1}(I_h(p) < C_h) = {\bf 1}(p \geq \bar{p}). \label{eq:ad-simple}\end{equation}

Combining \eqref{eq:unravel} and \eqref{eq:ad-simple} we see that if $I_h(p) < C_h$ then $p \geq \bar{p}$ and
\begin{equation} \sum_{\theta \in \S}{\bf 1}(p_h^\theta > p)M_h(\theta) = \sum_{\theta \in \S}{\bf 1}(p_h^\theta > p)(1- \sum_{h' \succ^\theta h} M_{h'}(\theta) ) = I_h(p),\label{eq:matched} \end{equation}
where the final inequality uses \eqref{eq:continuum-interest} and \eqref{eq:etaS}.

Meanwhile, if  $I_h(p) \geq C_h$, then $\bar{p} > p$, and \eqref{eq:unravel} and \eqref{eq:ad-simple} imply that
\begin{align} \sum_{\theta \in \S}{\bf 1}(p_h^\theta > p)M_h(\theta) &  = \sum_{\theta \in \S}{\bf 1}(p_h^\theta > \bar{p})(1- \sum_{h' \succ^\theta h} M_{h'}(\theta) ) + \sum_{\theta \in \S}{\bf 1}(p_h^\theta = \bar{p})(1- \sum_{h' \succ^\theta h} M_{h'}(\theta) ), \nonumber \\
& =  I_h(\bar{p}) + \sum_{\theta \in \S}{\bf 1}(p_h^\theta = \bar{p})(1- \sum_{h' \succ^\theta h} M_{h'}(\theta) ), \label{eq:final-step-prop1}
\end{align}
where the second line also follows from \eqref{eq:ad-simple}.

Note that because $\vdet(\lambda,C) \in \{0,1\}$, \eqref{eq:admissions} and \eqref{eq:enroll-prob} imply that $M_h(\theta) \in \{0,1\}$ and therefore \eqref{eq:continuum-interest} implies that $I_h(p) \in \mathbb{N}$. Furthermore, the fact that $\S$ has strict priorities implies that at discontinuities of $I_h$, it decreases by exactly one. We know from the definition of $\bar{p}$ that $I_h(\bar{p}) < C_h$ but $I_h(p) \geq C_h$ for all $p < \bar{p}$, so $I_h(\bar{p}) = C_h-1$ and $\sum_{\theta \in \S}{\bf 1}(p_h^\theta = \bar{p})(1- \sum_{h' \succ^\theta h} M_{h'}(\theta) ) = 1$. This implies that the expression in \eqref{eq:final-step-prop1} is equal to $C_h$, completing the proof.
\end{proof}

\begin{proof}[Proof of Proposition \tref{prop:finite}]
We first suppose that $(M,I,A)$ is $(\eta^\S,\vdet)$-stable, and show that $\mu^M$ has no blocking pairs. Note that the definition of $\vdet$ in \eqref{eq:vdet} implies that for all $\lambda\in \mathbb{R}_+, C \in \mathbb{N}$ we have $\vdet(\lambda,C) \in \{0,1\}$, so by \eqref{eq:admissions} and \eqref{eq:enroll-prob}, $M$ is deterministic and $\mu^M$ is well-defined. We now show that $\mu$ is feasible. Note that
\begin{align*}
 \sum_{\theta \in \S} M_h(\theta) & = \sum_{\theta \in \S} M_h(\theta){\bf 1}(p_h^\theta > 0) + \sum_{\theta \in \S} M_h(\theta){\bf 1}(p_h^\theta = 0) \\
 & \leq \sum_{\theta \in \S} M_h^\theta{\bf 1}(p_h^\theta > 0) + \eta^{\S}(\Theta_h(0)) A_h(I_h(0)) \\
 & = \min(I_h(0),C_h) + \eta^{\S}(\Theta_h(0)) {\bf 1}(I_h(0) < C_h) \\
 & \leq \min(I_h(0),C_h) + {\bf 1}(I_h(0) < C_h) \\
 & \leq C_h.
 \end{align*}
 The second line follows from $M_h(\theta) \leq A_h(p_h^\theta)$ (see \eqref{eq:enroll-prob}) and $\sum_{\theta \in \S} {\bf 1}(p_h^\theta = 0) = \eta^{\S}(\Theta_h(0))$ (see \eqref{eq:etaS}); the third from Lemma \ref{lem:enrollment-finite}, along with \eqref{eq:admissions} and \eqref{eq:vdet}, and the fourth because $\S$ has strict priorities.

Finally, we show that $\mu^M$ has no blocking pairs. By definition, if $h \succ^{\theta'} \mu^M(\theta')$ then $M_h(\theta') = 0$, and $A_h(p_h^{\theta'}) = 0$ by \eqref{eq:enroll-prob}. From this, \eqref{eq:admissions} and \eqref{eq:vdet} imply that $I_h(p_h^{\theta'})) \geq C_h$, so by Lemma \ref{lem:enrollment-finite}, 
\[|\{\theta \in \S : \mu^M(\theta) = h, p_h^\theta > p_h^{\theta'}\}| = \sum_{\theta \in \S} M_h(\theta){\bf 1}(p_h^\theta > p_h^{\theta'}) = \min(I_h(p_h^{\theta'}),C_h) = C_h.\]
That is, \eqref{eq:no-blocking} holds, so $\mu^M$ has no blocking pairs.

Next, we assume that $\mu$ is an $\S$-matching with no blocking pairs, and show that 
\begin{enumerate}[label=\roman*)]
\item $M^\mu$ ``agrees" with $\mu$: for $\theta \in \S, h \in \H_0$, we have 
\begin{equation} M^\mu_h(\theta) = {\bf 1}(\mu(\theta) = h). \label{eq:agreement}\end{equation}
\item $M^\mu$ is $(\eta^{\S}, \vdet)$-stable.
\end{enumerate}
We start by showing \eqref{eq:agreement}. Fix $\theta' \in \S$. Then for any $h \succ^{\theta'} \mu(\theta')$, the fact that $\mu$ has no blocking pairs implies that \eqref{eq:no-blocking} holds, from which the definition of $A^\mu$ in \eqref{eq:amu} implies that $A_h^\mu(p_h^{\theta'}) = 0$, so $M_h^\mu(\theta') = 0$. Meanwhile, for $h' = \mu(\theta')$, feasibility of $\mu$ implies 
\[| \{\theta \in \S : \mu(\theta) = h', p_{h'}^\theta > p_{h'}^{\theta'}\}| < | \{\theta \in \S : \mu(\theta) = h'\}| \leq C_{h'}.\]
Therefore, the definition of $A^\mu$ in \eqref{eq:amu} implies that $A_{h'}^\mu(p_{h'}^{\theta'}) = 1$, from which \eqref{eq:enroll-prob} implies that $M_{h'}(\theta') = 1$ (and that $M_h(\theta) = 0$ for all $h$ such that $\mu(\theta') \succ^{\theta'} h$). Thus, \eqref{eq:agreement} holds, implying that $\mu = \mu^{M^\mu}$.

We now show that $M^\mu$ is $(\eta^{\S}, \vdet)$-stable. Define $I^\mu = \mI(M^{\mu})$. Then we have
\begin{align}
I^\mu_h(p) & = \sum_{\theta \in \S} {\bf 1}(p_h^\theta > p)(1 - \sum_{h' \succ^\theta h} M_h^\mu(\theta)) \label{eq:ai-firststep} \\
& \geq \sum_{\theta \in \S} {\bf 1}(p_h^\theta > p) {\bf 1}(\mu(\theta) = h) = | \{\theta \in \S : p_{h}^\theta > p,\mu(\theta) = h\}|, \label{eq:ai-ineq}
\end{align}
where the first equality follows from \eqref{eq:continuum-interest} and the definition of $\eta^{\S}$ in \eqref{eq:etaS}, and the inequality from \eqref{eq:agreement}. We claim that for $h \in \H_0$, $p \in [0,1]$,
\begin{equation} A^\mu_h(p) = {\bf 1}(I^\mu_h(p) < C_h) = \vdet(I^\mu_h(p),C_h), \label{eq:astable}\end{equation}
implying that $A^\mu = \mA(I^\mu) = \mA(\mI(\M(A^\mu)))$, so $A^\mu$ is $(\eta^{\S}, \vdet)$-stable, and therefore so is $M^\mu$. To establish \eqref{eq:astable}, we show that $A^\mu_h(p) = 0$ implies $I_h^\mu(p) \geq C_h$, and $A^\mu_h(p) = 1$ implies $I_h^\mu(p) < C_h$.

If $A^\mu_h(p) = 0$, then by definition of $A^\mu$ in \eqref{eq:amu}, 
\[ | \{\theta \in \S : p_{h}^\theta > p,\mu(\theta) = h\}| \geq C_h. \]
By \eqref{eq:ai-ineq}, this implies that $I_h^\mu(p) \geq C_h$. Conversely, if $A_h^\mu(p) = 1$, then by definition
\[ | \{\theta \in \S : p_{h}^\theta > p,\mu(\theta) = h\}| < C_h. \]
This implies that for each $\theta \in \S$ that contributes to the sum in \eqref{eq:ai-firststep}, $\mu(\theta) = h$ (otherwise, \eqref{eq:no-blocking} would be violated). Therefore, the inequality in \eqref{eq:ai-ineq} is an equality, implying that $I_h^\mu(p) < C_h$.
\end{proof}

\subsection{Proof of Proposition \tref{prop:al}} \label{app:prop-al}

We start by establishing an analog to Lemma \ref{lem:enrollment-finite}, which says that for any $(\eta,\mv)$-stable outcome, the measure of students matched to school $h$ can be determined by the measure of interest in $h$. 
\begin{lemma} \label{lem:enrollment}
If $\eta$ has strict priorities and $\mV$ is weakly decreasing in its first argument, then for any $(\eta, \mV)$-stable outcome $(M,\I,A)$, any school $h \in \H$, and any $p \in [0,1]$,
\[ \int {\bf 1}(p_h^\theta > p) M_h(\theta) d\eta(\theta)=  \int_{0}^{\I_h(p)} \mV(\lambda,C_h) d \lambda.\]
\end{lemma}
If $\mV = \vdet$, then the expression on the right is $\min(\I_h(p), C_h)$, matching that in Lemma \ref{lem:enrollment-finite}. However, Lemma \ref{lem:enrollment} provides a more general expression that holds for any monotone vacancy function.\todo{\footnote{\na{The result does not hold when $\eta$ is a discrete measure. To see this, consider a simple case with one school with capacity one, and two students, with priority $p_H > p_L$. Regardless of what matching we start with, both students are interested in the school, so the resulting interest function is $I(p) = {\bf 1}(p \le p_H) + {\bf 1}(p \le p_L)$ (a step function that steps down at $p_L$ and $p_H$). For this interest function, if we use $\vpois$, the resulting admissions function is $A(p) = 1$ for $p \in [p_H,1]$, $A(p) = 1/e$ for $p \in [p_L,p_H)$ and $A(p) = 1/e^2$ for $p \in [0,p_L)$. Therefore, the top student gets in for sure, and the lower student with probability $1/e$. The total number of matches predicted by the stable matching $M$ is $1 + 1/e$; this is more than $\int_0^\infty \vpois(\lambda,1) d\lambda = 1$.}}}

\begin{proof}[Proof of Lemma \ref{lem:enrollment}]
%We note that if $\eta$ is a continuum measure with strict priorities, then \eqref{eq:continuum-interest} implies that $\I$ is continuous and decreasing. Therefore, it is differentiable almost everywhere. The essence of the proof is simple: 
%\begin{align}
%\int {\bf 1}(p_h^\theta > p) M_h(\theta) d\eta(\theta) & = \int {\bf 1}(p_h^\theta > p) A_h(p_h^\theta) \prod_{h' \succ^\theta h} (1 - A_{h'}(p_{h'}^\theta)) d\eta(\theta)\nonumber  \\
%%& = \int_{p}^1 A_h(q) \int_{\Theta_h^q} \prod_{h' \succ^\theta h} (1 - A_{h'}(p_{h'}^\theta)) d\eta(\theta) dq \nonumber \\
%& = \int_{p}^1 \mV(\I_h(q),C_h) \I_h'(q) dq \nonumber \\
%& = \int_{0}^{\I_h(p)} \mV(\lambda,C_h) d \lambda. \nonumber
%\end{align}
%The first equality follows from \eqref{eq:enroll-prob}, the second from \eqref{eq:continuum-interest} and \eqref{eq:admissions}, and the last from substitution.
%
%\na{Proof starts here.}
Fix $n \in \mathbb{N}$, define $m = \lceil nI_h(p) \rceil$, and define $\{L_i\}_{i = 0}^{m}$ by $L_i = i/n$ for $i < m$, and $L_m = I_h(p)$. Note that if $\eta$ is a continuum measure with strict priorities, then \eqref{eq:continuum-interest} implies that $\I$ is continuous and decreasing. In particular, this implies that we can choose $1 = P_0 > P_1 > \cdots P_m = p$ such that $I_h(P_i) = L_i$ for each $i$. We claim the following chain of inequalities:
\begin{align}
\int {\bf 1}(p_h^\theta > p) M_h(\theta) d\eta(\theta) & = \int {\bf 1}(p_h^\theta > p) A_h(p_h^\theta) \prod_{h' \succ^\theta h} (1 - A_{h'}(p_{h'}^\theta)) d\eta(\theta)\nonumber  \\
& = \int {\bf 1}(p_h^\theta > p) \mV(I_h(p_h^\theta),C_h) \prod_{h' \succ^\theta h} (1 - A_{h'}(p_{h'}^\theta)) d\eta(\theta) \nonumber\\
& = \int \sum_{i = 1}^m {\bf 1}(P_{i-1} \geq p_h^\theta > P_{i})  \mV(I_h(p_h^\theta),C_h) \prod_{h' \succ^\theta h} (1 - A_{h'}(p_{h'}^\theta)) d\eta(\theta)\nonumber\\
& \geq  \sum_{i = 1}^m \mV(L_i,C_h) \int {\bf 1}(P_{i-1} \geq p_h^\theta > P_{i}) \prod_{h' \succ^\theta h} (1 - A_{h'}(p_{h'}^\theta)) d\eta(\theta). \label{eq:first-step}
\end{align}
The first equality holds from \eqref{eq:enroll-prob}, the second from \eqref{eq:admissions}, and the third by definition of $P_i$. The final inequality comes from the fact that $\mV$ is weakly decreasing in its first argument and $I_h$ is weakly decreasing, and thus $p_h^\theta > P_i$ implies  $\mV(I_h(p_h^\theta),C_h) \geq \mV(I_h(P_i),C_h) = \mV(L_i,C_h)$. Note that 
\begin{align}
 \int {\bf 1}(P_{i-1} \geq \,\,& p_h^\theta > P_{i}) \prod_{h' \succ^\theta h} (1 - A_{h'}(p_{h'}^\theta)) d\eta(\theta)\nonumber \\
  & = \int{\bf 1}(p_h^\theta > P_{i}) \prod_{h' \succ^\theta h} (1 - A_{h'}(p_{h'}^\theta)) d\eta(\theta) - \int {\bf 1}(p_h^\theta > P_{i-1}) \prod_{h' \succ^\theta h} (1 - A_{h'}(p_{h'}^\theta)) d\eta(\theta) \nonumber\\
  & = I_h(P_i) - I_h(P_{i-1})\nonumber \\
  & = L_i - L_{i-1}. \label{eq:second-step}
 \end{align}
 The second equality follows from \eqref{eq:continuum-interest} and the third from the choice of $P_i$. Combining \eqref{eq:first-step} and \eqref{eq:second-step}, and noting that $L_{i} - L_{i - 1} = 1/n$ for $i < m$, we get 
 \[\int {\bf 1}(p_h^\theta > p) M_h(\theta) d\eta(\theta) \geq \frac{1}{n} \sum_{i = 1}^{m-1} \mV(L_i,C_h). \]
This holds for any $n \in \mathbb{N}$. Taking the limit as $n \rightarrow \infty$ yields
\begin{equation} \int {\bf 1}(p_h^\theta > p) M_h(\theta) d\eta(\theta) \geq \int_0^{I_h(p)} \mV(\lambda,C_h)d\lambda. \label{eq:done} \end{equation}
The inequality in \eqref{eq:first-step} can be reversed if we replace $\mV(L_i,C_h)$ with $\mV(L_{i-1},C_h)$. From there, an analogous argument implies \eqref{eq:done} with the inequality reversed, completing the proof.

\end{proof}

\begin{proof}[Proof of Proposition \tref{prop:al}]
\todo{Could take a pass to add more intuition/explanation.} Given cutoffs $P \in [0,1]^\H$, we define
\begin{align}
 A^P_h(p) & =  {\bf 1}( p > P_h).\label{eq:aP} \\
 M^P & = \M(A^P). \label{eq:mP}
 \end{align}

We begin by noting two equalities that will repeatedly prove useful. Note that the definition of $\mA$ through \eqref{eq:admissions} and \eqref{eq:vdet}, and the definition of $A^P$ and $\P(I)$ in \eqref{eq:aP} and \eqref{eq:cutoffs} imply that for any interest function $I \in \Interest$,
\begin{equation} \hspace{.3 in} \mA(I) = A^{\P(I)}. \label{eq:a-equality}\end{equation}
Furthermore, the definition of $I^P$, $A^P$ and $M^P$ in \eqref{eq:iP}, \eqref{eq:aP} and \eqref{eq:mP}, along with \eqref{eq:continuum-interest} and \eqref{eq:enroll-prob}, imply that for any cutoff vector $P$, 
\begin{equation} I^P = \mI(\M(A^P)). \label{eq:i-equality}\end{equation}

We first assume that $I$ is $(\eta,\vdet)$-stable, and show that $I = I^{\P(I)}$. 
By Definition \ref{def:stability},  if $I$ is $(\eta,\vdet)$-stable, then so is
\begin{equation} \M(\mA(I)) = \M(\mA^{\P(I)}) =M^{\P(I)}, \label{eq:m-coincide} \end{equation}
where the equalities follow from \eqref{eq:a-equality} and \eqref{eq:mP}, respectively. Furthermore, stability of $I$ implies that
\begin{equation} I = \mI(\M(\mA(I))) = \mI(\M(\mA^{\P(I)})) = I^{\P(I)}, \label{eq:i-coincide}\end{equation}
where the second and third equalities follow from \eqref{eq:m-coincide} and \eqref{eq:i-equality}, respectively. 

Next, we show that $\P(I)$ is market-clearing. We claim that
\begin{equation} D_h(\P(I)) = \int M_h^{\P(I)}(\theta) d\eta(\theta) = \int_0^{I_h(0)} \vdet(\lambda,C_h) d\lambda = \min(I_h(0),C_h) \leq C_h. \label{eq:demand-chain} \end{equation}
The first equality follows from \eqref{eq:demand}, the second from Lemma \ref{lem:enrollment} and the fact that $\mI(M^\P(I)) = I$ (by \eqref{eq:m-coincide} and \eqref{eq:i-coincide}), and the third from the definition of $\vdet$ in \eqref{eq:vdet}. Furthermore, if $\P_h(I) > 0$, it follows from definition of $\P$ in \eqref{eq:cutoffs} that $I_h(P_h) \geq C_h$ (this also uses the fact that $I_h$ is continuous, which follows from \eqref{eq:continuum-interest} and the fact that $\eta$ is a continuum measure with strict priorities). Because \eqref{eq:continuum-interest} implies that $I_h$ is weakly decreasing, it follows that $I_h(0) \geq C_h$, and therefore the inequality in \eqref{eq:demand-chain} is tight. Therefore, $\P(I)$ is market-clearing.

Finally, we show that if $P \in [0,1]^\H$ is market-clearing, then $I^P$ is $(\eta,\vdet)$-stable. By \eqref{eq:a-equality} and  \eqref{eq:i-equality},
\begin{equation}\mI(\M(\mA(I^P))) = \mI(\M(A^{\P(I^P)})) = I^{\P(I^P)}. \label{eq:partway-to-stability} \end{equation}
We wish to show that this is equal to $I^P$. For less cumbersome notation, we define the cutoff vector 
\begin{equation} \tilde{P} = \P(I^P). \label{eq:tildeP} \end{equation}
The steps to prove that $I^{\tilde{P}} = I^P$ are as follows:
\begin{enumerate}[label=\Roman*.]
\item \label{step:one} Show that $I_h^P(P_h) = D_h(P)$. Conclude that 
\begin{equation} \tilde{P}_h \geq P_h.  \label{eq:p-ineq}\end{equation}
\item \label{step:two} Define 
\begin{equation} \Delta_h = \{ \theta : p_h^\theta \in (P_h,\tilde{P}_h], p_{h'}^\theta \leq P_{h'} \text{ for all } h' \succ^{\theta} h \}  \label{eq:delta-h}\end{equation}
to be the $\eta$-measure of students who have priority between $P_h$ and $\tilde{P}_h$ at $h$, and are not admitted to any school preferred to $h$ under either $P$ or $\tilde{P}$. Establish that $I_h^P$ is constant on $(P_h,\tilde{P}_h]$, and therefore that
\begin{equation} 0 = I_h^P(P_h) - I_h^P(\tilde{P}_h) = \eta(\Delta_h).\label{eq:meas-zero} \end{equation}
%\[ \eta(\{ \theta : p_h^\theta \in (P_h, \tilde{P}_h], p_{h'}^\theta \leq P_{h'} \text{ for all } h' \succ^\theta h \}) = 0.\]
\item \label{step:three} Conclude that for all $h \in \H$ and $p \in [0,1]$, 
\begin{equation} I^{\tilde{P}}_h(p) = I_h^P(p). \label{eq:no-difference} \end{equation}
\end{enumerate}
We now establish step \ref{step:one} Note that
\begin{align*} 
I_h^P(p) & = \int {\bf 1}(p_h^\theta > p)(1 - \sum_{h' \succ^\theta h} M_{h'}^P(\theta)) d\eta(\theta)\\
& = \int {\bf 1}(p_h^\theta > p)\prod_{h' \succ^\theta h}(1 -  A_{h'}^P(p_{h'}^\theta)) d\eta(\theta),
\end{align*} 
where the first line follows from \eqref{eq:i-equality} and \eqref{eq:continuum-interest}, and the second from \eqref{eq:unravel}. It follows that 
\begin{align}
I_h^P(P_h) & =  \int A_h^P(p_h^\theta)\prod_{h' \succ^\theta h}(1 -  A_{h'}^P(p_{h'}^\theta)) d\eta(\theta)\nonumber \\
& = \int M_h^P d\eta(\theta) =  D_h(P),\label{eq:ieqd}
\end{align}
where the second line follows from \eqref{eq:enroll-prob} and the definitions of $A^P$ and $D_h(P)$ in \eqref{eq:aP} and \eqref{eq:demand}. From \eqref{eq:ieqd} and the definition of $\P(\cdot)$ in \eqref{eq:cutoffs}, \eqref{eq:p-ineq} follows.

Next, we move to step \ref{step:two} Because $\eta$ is a continuous measure with strict priorities, \eqref{eq:continuum-interest} implies that $I^P_h$ is continuous for each $h \in \H_0$. Therefore, the definition of $\P(\cdot)$ in \eqref{eq:cutoffs} implies that either $\tilde{P}_h = 0$ (in which case \eqref{eq:p-ineq} implies that $P_h = \tilde{P}_h$), or $I^P_h(\tilde{P}) = C_h$. But then \eqref{eq:p-ineq} and the fact that $I^P$ is decreasing imply that 
\[I_h^P(P_h) \geq I_h^P(\tilde{P}_h) = C_h.\]
Because $P$ is market-clearing, $D_h(P) \leq C_h$, implying that the inequality above must hold with equality. Therefore, $I_h^P$ is constant on $(P_h,\tilde{P}_h]$. In particular, applying the definition of $I^P$ in \eqref{eq:iP} reveals that \eqref{eq:meas-zero} holds.

Finally, we move to step \ref{step:three} By \eqref{eq:iP}, for any $h \in \H$ and $p \in [0,1]$ we have
\begin{equation} I_h^{\tilde{P}}(p) - I_h^P(p)  = \eta( \Delta), \label{eq:difference-delta}\end{equation}
where
\[ \Delta = \{\theta : p_h^\theta > p, \prod_{h' \succ^\theta h} {\bf 1}(p_{h'}^\theta \leq \tilde{P}_{h'}) - \prod_{h' \succ^\theta h} {\bf 1}(p_{h'}^\theta \leq P_{h'}) = 1  \}\]
That is, the difference $I_h^{\tilde{P}}(p) - I_h^P(p)$ is the measure of students who have priority above $p$ at $h$, and are admitted to a school preferred to $h$ under cutoffs $P$, but are not admitted to any such school under $\tilde{P}$. For any $\theta \in \Delta$, there is some most-preferred school $h'$ where $\theta$ is admitted under $P$ but not under $\tilde{P}$. Then \eqref{eq:delta-h} implies that $\theta \in \Delta_{h'}$. In other words, $\Delta \subseteq \bigcup_{h \in \H} \Delta_h$. From this, \eqref{eq:meas-zero} implies that $\eta(\Delta) = 0$, and \eqref{eq:difference-delta} implies that \eqref{eq:no-difference} holds. This establishes that $I^P = I^{\tilde{P}} = I^{\P(I^P)}$, from which \eqref{eq:partway-to-stability} implies that $I^P$ is stable.

%Note that by the definition of $\mA$ (through \eqref{eq:admissions} and \eqref{eq:vdet}), and the definition of $A^P$ and $\P(I)$ in \eqref{eq:aP} and \eqref{eq:cutoffs}, for every $I$ we have $\mA(I) = A^{\P(I)}$. Therefore, if $(M,I,A)$ is a $(\eta,\vdet)$-stable outcome, then 
%\begin{equation} M^{\P(I)} = \M(\mA^{\P(I)}) =  \M(\mA(I)) = M. \label{eq:m-coincide} \end{equation}
%We now show that $\P(I)$ is a market-clearing cutoff. Note that
%\begin{equation} D_h(\P(I)) = \int M_h(\theta) d\eta(\theta) = \int_0^{I_h(0)} \vdet(\lambda,C_h) d\lambda = \min(I_h(0),C_h) \leq C_h. \label{eq:demand-chain} \end{equation}
%The first equality follows from \eqref{eq:demand} and \eqref{eq:m-coincide}, the second from Lemma \ref{lem:enrollment}, and the third from the definition of $\vdet$ in \eqref{eq:vdet}. Furthermore, if $\P_h(I) > 0$, it follows from definition of $\P$ in \eqref{eq:cutoffs} that $I_h(P_h) \geq C_h$ (this also uses the fact that $I_h$ is continuous, which follows from \eqref{eq:continuum-interest} and the fact that $\eta$ is a continuum measure with strict priorities). Because \eqref{eq:continuum-interest} implies that $I_h$ is weakly decreasing, it follows that $I_h(0) \geq C_h$, and therefore the inequality in \eqref{eq:demand-chain} is tight. Therefore, $\P(I)$ is market-clearing.

\end{proof}

\subsection{Proof of Theorems \tref{thm:rht} and \tref{thm:uniqueness}} \label{app:main-results}

\begin{proof}[Proof of Theorem \tref{thm:rht}] 
By Theorem \tref{thm:existence}, there exist maximal and minimal stable outcomes, corresponding to the school-optimal and student-optimal stable outcomes, respectively. Denote these outcomes by $(M^H, \I^H,A^H) \succeq (M^L, \I^L,A^L)$, respectively. It is enough to prove the result for these outcomes. Note that 
\begin{align}
\int \sum_{h \succ^\theta \emptyset} M_h^L(\theta) d\eta(\theta) &  = \sum_{h \in \H} \int M_h^L(\theta) d\eta(\theta) = \sum_{h \in \H} \int_0^{\I_h^L(0)} \mV(\lambda,C_h) d\lambda \nonumber \\
&  \geq \sum_{h \in \H} \int_0^{\I^H_h(0)} \mV(\lambda,C_h) = \sum_{h \in \H} \int M^H_h(\theta) d\eta(\theta) = \int \sum_{h \succ^\theta \emptyset} M^H_h(\theta) d\eta(\theta). \label{eq:geq} \end{align}
The first and last equalities hold because $A_\emptyset(p) = 1$ for all $p$, so by \eqref{eq:enroll-prob}, $M_h(\theta) = 0$ if $\emptyset \succ^\theta h$. The second and second-to-last equalities hold by Lemma \ref{lem:enrollment}. The inequality follows from the fact that $\I^L \succeq^\Interest \I^H$. But $M^H \succeq^\Matching M^L$ implies 
\begin{equation} \sum_{h \succ^\theta \emptyset} M_h^L(\theta) \leq \sum_{h \succ^\theta \emptyset} M_h^H(\theta) \hspace{.5 in } \forall \theta \in \Theta. \label{eq:leq} \end{equation} 
Therefore, the inequality in \eqref{eq:geq} must hold with equality. In particular, this implies that \eqref{eq:school-rht} holds for each $h \in \H$. Furthermore, this implies that \eqref{eq:student-rht} holds for all $\theta$ except possibly a set of $\eta$-measure zero. \todo{Show that this holds for all students (not even measure zero exception).}
%By definition, 
%\begin{equation} A_h^H(p) \leq A_h^L(p) \text{ for all } h \in \H, p \in [0,1]. \label{eq:aineq} \end{equation}
%That is, 
%\begin{equation} \matchedstudents(A^H) =  \matchedstudents(A^L)  \label{eq:ruralhospital} \end{equation}
%By \eqref{eq:thetamatch} and \eqref{eq:aineq} we have
%\begin{equation} \M^\theta(A^H) \leq \M^\theta(A^L) \text{ for all }\theta \in \Theta. \label{eq:thetamatchineq} \end{equation} 
%By \eqref{eq:matchedstudents}, this immediately implies that $\matchedstudents(A^H) \leq \matchedstudents(A^L)$. Meanwhile, the function $\filledseats$ defined in \eqref{eq:filledseats} is weakly increasing. From this and Lemma \ref{lem:matchconsistency} we have that 
%\[\filledseats(\I^L) \leq \filledseats(\I^H) = \matchedstudents(A^H) \leq \matchedstudents(A^L),\]
%so the inequalities above must hold with equality, implying that \eqref{eq:ruralhospital} holds.
%\na{From this, it immediately follows by definition of \eqref{eq:filledseats} and \eqref{eq:matchedstudents} that $\enrollment(\I_h(0),C_h)$ is identical for all $h$ ($I_h$ may not be if $\enrollment$ is not strictly increasing -- i.e. if $\mV$ drops to zero). Also implies that student match rate is identical outside of a set of measure zero. May be able to get to identical for all students, but this would take more thought. Also note that this implies result for any pair of stable matches.}
\end{proof}

\begin{proof}[Proof of Theorem \tref{thm:uniqueness}] 
It suffices to show that there is a unique stable interest function: that is, if $\I^H$ and $\I^L$ are the largest and smallest stable interest functions according to $\succeq^{\Interest}$, then $\I^H = \I^L$. We let $A^H, M^H$ be the admissions function and matching associated with $\I^H$, and define $A^L, M^L$ analogously.  We note that by \eqref{eq:admissions} and the fact that $\mV$ is decreasing in its first argument, $\I^H \succeq^\Interest \I^L$ implies that 
\begin{equation} A^L \succeq^{\Admissions} A^H.  \label{eq:aineq}\end{equation}
The proof proceeds by contradiction, showing that $\I^H \succ^{\interest} \I^L$ implies that Theorem \tref{thm:rht} does not hold. That is, $\I^H \succ^{\interest} \I^L$ implies the existence of a set $\tilde{\Theta}$ with $\eta(\tilde{\Theta}) > 0$ such that %\todo{read through and make sure this makes sense -- what does $M^H(\theta)$ mean?}
\begin{equation} \sum_{h \in \H} M_h^H(\theta) <  \sum_{h \in \H} M_h^L(\theta)\text{ for all  }\theta \in \tilde{\Theta}. \label{eq:thetamatchineq2}\end{equation} 
We establish existence of such a $\tilde{\Theta}$ in three steps.
\begin{enumerate}[label = \Roman*.]
\item Note that Definition \ref{def:strict} and \eqref{eq:continuum-interest} imply that 
\begin{enumerate}[label = \alph*)]
\item the stable interest functions $\I^H$ and $\I^L$ are component-wise continuous, and
\item for all $h \in \H$, $\I_h^H(1) = \I_h^L(1) = 0$.
\end{enumerate}
\item \label{item:secondstep} These jointly imply that if $\I^H_h(p) > \I^L_h(p)$ for some $p \in [0,1]$ and $h \in \H$, then there must exist an interval $(\underline{p}, \overline{p})$ such that:
\begin{enumerate}[label = \alph*)]
\item $\I_h^H(p) > \I_h^L(p)$ for $p \in (\underline{p}, \overline{p})$, and \label{item:strict}
\item $\I_h^H(\overline{p}) > \I_h^H(\underline{p})$. \label{item:strict2}
\end{enumerate}
That is, $\I_h^H$ is not constant and strictly larger than $\I_h^L$ on this interval.
\item Define 
\begin{equation}S = \{\theta : M_\emptyset^H(\theta) = 0\} \} \label{eq:S} \end{equation}
to be the set of student types who are sure to be matched. Define
\begin{equation} \tilde{\Theta} = \{ \theta : h \succ^\theta \emptyset, p_h^\theta \in (\underline{p},\overline{p})\} \backslash S. \label{eq:thetatilde} \end{equation}
%\[ \tilde{\Theta} = \{ \theta : h \succ^\theta \emptyset, p_h^\theta \in (\underline{p},\overline{p}), \prod_{h' \succ^\theta h} (1 - A_{h'}^H(p_{h'}^\theta)) > 0\}.\]
We will show that \eqref{eq:thetamatchineq2} holds, and that $\eta(\tilde{\Theta}) > 0$.
\end{enumerate}
To see that \eqref{eq:thetamatchineq2} holds, note that if $\theta \in \tilde{\Theta}$,
\begin{equation} A_h^H(p_h^\theta)  = \mV(\I^H(p_h^\theta),C_h) < \mV(\I^L(p_h^\theta),C_h) = A_h^L(p_h^\theta), \label{eq:astrictineq} \end{equation}
where the equalities hold by \eqref{eq:admissions} and the inequality follows from \ref{item:secondstep}\ref{item:strict}, the fact that $p_h^\theta \in (\underline{p},\overline{p})$, and the fact that $\mV(\cdot, C_h)$ is strictly decreasing. Thus, when comparing $A^H$ to $A^L$, each student in $\tilde{\Theta}$ is
\begin{enumerate}[label=\roman*.]
\item weakly less likely to be admitted to each school under $A^H$ by \eqref{eq:aineq},
\item strictly less likely to be admitted to $h$ under $A^H$ by \eqref{eq:astrictineq}, and
\item not certain to be admitted to any school by definition of $\tilde{\Theta}$ in \eqref{eq:thetatilde}.
\end{enumerate}
From this, \eqref{eq:enroll-prob} implies that each $\theta \in \tilde{\Theta}$ is strictly less likely to match under $A^H$. That is, \eqref{eq:thetamatchineq2} holds.

Finally, we show that $\eta(\tilde{\Theta}) > 0$. By definition of $\tilde{\Theta}$ in \eqref{eq:thetatilde}
\begin{align}
\eta(\tilde{\Theta} \cup S) & \geq \int {\bf 1}(h \succ^\theta \emptyset, p_h^\theta \in (\underline{p},\overline{p})) d\eta(\theta) \label{eq:almost} \\
& \geq \int {\bf 1}(h \succ^\theta \emptyset, p_h^\theta \in (\underline{p},\overline{p}))(1 - \sum_{h' \succ^\theta h} M_{h'}^H(\theta)) d\eta(\theta)\nonumber  \\
& = \I^H_h(\underline{p}) -  \I^H_h(\overline{p}) \nonumber \\
& > 0,\nonumber
%\E_{\theta \sim \eta} \left[ {\bf 1}(h \succ^\theta \emptyset, p_h^\theta > \underline{p}) {\bf 1}\left(\prod_{h' \succ^\theta h} (1 - A_{h'}^H(p_{h'}^\theta)) > 0\right)\right] 
\end{align}
where the equality follows from \eqref{eq:continuum-interest} and stability of $\I^H$, and the final line follows from \ref{item:secondstep}\ref{item:strict2}. We complete the proof by showing that $\eta(S) = 0$.

Because $\mV(\cdot, C_h)$ is strictly decreasing, \eqref{eq:admissions} implies that $A_h^H(p) = 1$ if and only if $\I^H_h(p) = 0$. Define $p_h  = \inf \{p : \I^H_h(p_h^\theta) = 0\}$ to be the lowest priority at school $h$ that guarantees admission, and note that
 \begin{equation} 0 = \I_h^H(p_h) = \int {\bf 1}(p_h^\theta > p_h)(1 - \sum_{h' \succ^\theta h} M_{h'}^H(\theta)) d\eta(\theta) , \label{eq:lambdazero} \end{equation}
where the second equality comes from stability of $\I^H$. Define 
\[ S_h = \{\theta :  p_h^\theta > p_h, \sum_{h' \succ^\theta h} M_{h'}^H(\theta) < 1 \}\] 
to be the set of agents who are certain to be admitted to $h$, and not certain to be admitted to any option that they prefer to $h$. Note that 
 \[ \eta(S_h) = \int {\bf 1}( p_h^\theta > p_h){\bf 1} \left( \sum_{h' \succ^\theta h} M_{h'}^H(\theta) < 1 \right)  d\eta(\theta) = 0,\]
 where the second equality follows from \eqref{eq:lambdazero}. Because $S = \bigcup_{h \in \H} S_h$, it follows that $\eta(S) = 0$ and thus $\eta(\tilde{\Theta})>0$ by \eqref{eq:almost}.

\end{proof}

\section{Proofs from Section \ref{sec:poisson}}
\na{Why can't I make this sref?}

Define 
\begin{equation} \acceptancerate(\lambda,C) = \enrollment(\lambda,C)/\lambda = \frac{1}{\lambda} \int_0^\lambda \mV(x,C) dx,\end{equation}
where the second equality follows from the definition of $\enrollment$ in \eqref{eq:enrollment}. The following result implies that $\acceptancerate$ is decreasing in its first argument.

\begin{lemma} \label{lem:monotonicity}
Given $f : \mathbb{R}_+ \rightarrow \mathbb{R}_+$, define $g : \mathbb{R}_+ \rightarrow \mathbb{R}_+$ by 
\[ g(y) = \frac{1}{y} \int_0^y f(x) dx.\]
If $f$ is weakly increasing, then so is $g$. If $f$ is weakly decreasing, then so is $g$.
\end{lemma}
\begin{proof}[Proof of Lemma \ref{lem:monotonicity}]
Note that 
\[g'(y) = \frac{yf(y) - \int_0^y f(x) dx}{y^2}.\]
If $f$ is weakly increasing, then $yf(y) \geq \int_0^y f(x)dx$; if $f$ is weakly decreasing, the inequality reverses.
\end{proof}

\begin{proof}[Proof of Proposition \tref{prop:more-seats}]
\todo{Should add language interpreting/explaining ideas.}
Let $\ell(\theta) = | \{h : h \succ^\theta \emptyset\}|$ be the number of schools listed by type $\theta$. We note that for any individually rational matching $M$ and any $\theta$, %\na{are we using that $M$ is IR?}
\begin{align}
 \sum_{h' \succ^\theta \emptyset} M_{h'}(\theta)R_{h'}(\theta) &  \leq \ell(\theta)(1 - \sum_{h' \succ^\theta \emptyset} M_{h'}(\theta)) + \sum_{h' \succ^\theta \emptyset} M_{h'}(\theta)R_{h'}(\theta) \nonumber \\
 & = \ell(\theta) - \sum_{h' \succ^\theta \emptyset} M_{h'}(\theta)(\ell(\theta) - R_{h'}(\theta))  \nonumber \\
 & = \ell(\theta) - \sum_{h' \succ^\theta \emptyset} \sum_{h' \succ^\theta h \succ^\theta \emptyset} M_{h'}(\theta)  \nonumber\\
 & = \sum_{h \succ^\theta \emptyset} (1 - \sum_{h' \succ^\theta h} M_{h'}(\theta)) \\
 & = \sum_{h \in \H \emptyset} (1 - \sum_{h' \succ^\theta h} M_{h'}(\theta)) \label{eq:interest-ub}
 \end{align}
 where the third line follows from the fact that $\ell(\theta) - R_{h'}(\theta)$ is the number of acceptable schools that rank below $h'$, the fourth follows by exchanging the order of summation, and the last uses the fact that $M$ is individually rational.
 
 From \eqref{eq:interest-ub} and the definition of $\averagerank$ in \eqref{eq:avg-rank}, it follows that if $(M,I,A)$ is a stable outcome, 
 \begin{equation} \averagerank(M) \leq \frac{\int \sum_{h \in \H} M_h(\theta) (1 - \sum_{h' \succ^\theta h} M_{h'}(\theta)) d\eta(\theta)}{\int \sum_{h \in \H} M_h(\theta) d\eta(\theta)} = \frac{\sum_{h \in \H} \I_h(0)}{\sum_{h \in \H} \int_0^{I_h(0)} \mV(\lambda,C_h) d\lambda}. \label{eq-avg-rank-bound}  \end{equation} 
 Note that the final equality follows by the fact that $M$ is stable, \eqref{eq:continuum-interest} and Lemma \ref{lem:enrollment}.  In a symmetric iid market, $C_h = C_{h'}$ and $I_h = I_{h'}$ for all $h, h' \in \H$ \todo{Check claim that $I_h = I_{h'}$ -- does this require iid?}, so \label{eq:avg-rank-bound} implies that for any $h \in \H$, 
\begin{equation}  \averagerank(M) \leq \frac{\I_h(0)}{\int_0^{I_h(0)} \mV(\lambda,C_h) d\lambda} = \frac{1}{\acceptancerate(I_h(0),C_h)}. \label{eq:avg-rank-bound-2}\end{equation} 
 %\eta(\Theta) \geq 
Lemma \ref{lem:monotonicity}  implies that $\acceptancerate$ is decreasing in its first argument, so we can obtain an upper bound on this expression by obtaining an upper bound on $I_h(0)$. But it is clear that the denominator in \eqref{eq-avg-rank-bound} is at most $\eta(\Theta)$, from which symmetry implies that $\enrollment(I_h(0),C_h) \leq \rho = \enrollment(\Lambda(\rho,C_h),C_h)$. Because $\enrollment$ is increasing in its first argument, this implies that $I_h(0) \leq \Lambda(\rho,C)$. Plugging this into \eqref{eq:avg-rank-bound-2} completes the proof.
\end{proof}

\begin{lemma} \label{lem:ar}
The function $AR : (0,1] \rightarrow \mathbb{R}_+$ defined by\todo{\footnote{\todo{For a general list length distribution, the appropriate definition is \[AR(\alpha) = \frac{\mu(\alpha) - \E[\ell(1 - \alpha)^\ell]}{\E[1 - (1-\alpha)^\ell]} = 1 -  \frac{(1-\alpha) \mu'(\alpha)}{\mu(\alpha)}.\]}}}
\begin{equation} AR(q) = \frac{1}{q} - \frac{\ell (1 - q)^\ell}{1 - (1 - q)^\ell},\end{equation}
is decreasing in $q$.
\end{lemma}

\todo{
\begin{proof}[Proof of Lemma \ref{lem:ar}]
Complete this.
\end{proof}
}

%\begin{equation} \sum_{k = 1}^\ell k q (1-q)^{k-1} =  (1-(1-q)^\ell) AR(q). \label{eq:summation}\end{equation}

%\na{
%We claim that $q$ can be rewritten as \na{could also state in terms of $\acceptancerate$.}
%\begin{equation}
%q = \frac{1}{\lambda} \int_0^{I(0)} \mV(\lambda,C) d\lambda.
%\end{equation} 
%To see this, note that in an iid market, for each preference profile $\succ$, the priorities $p_h$ are drawn uniformly and independently. Therefore, \eqref{eq:continuum-interest} implies that $I_h(p) = I_h(0)(1-p)$ for each $h$.
%}

\begin{proof}[Proof of Proposition \tref{prop:more-students}]
Let $(M,I,A)$ be the unique $(\eta^{IID}, \vpois)$-stable outcome. Note that by symmetry, $I_h(p) = I_{h'}(p)$ and $A_h(p) = A_{h'}(p)$ for all $h, h' \in \H$ and $p \in [0,1]$, so in what follows, we write $I(p)$ and $A(p)$ in place of $I_h(p)$ and $A_h(p)$. \na{Do we want to impose this right away, or hold off for as long as possible? Former seems clearer, but latter generalizes more easily.}

Define
\begin{equation} q = \int_0^1 \vpois(I(p),C) dp. \end{equation}

Note that 
\begin{align}
 \int \sum_{h \in \H} M_h(\theta) d\eta(\theta) & = \int (1 - \prod_{h \succ^\theta \emptyset} (1 - A(p_h^\theta)) d\eta^{IID}(\theta) \nonumber \\
 & = \int (1 - \prod_{h \succ^\theta \emptyset} (1 - \vpois(I(p_h^\theta),C)) d\eta^{IID}(\theta) \nonumber \\
 & = \eta^{IID}(\Theta) (1 - (1 - q)^\ell), \label{eq:ar-num}
\end{align}
where the first equality follows from \eqref{eq:enroll-prob}, the second from \eqref{eq:admissions}, and the last from the fact that we assume all students list $\ell$ schools, and in an iid market, the priorities $p_h$ are drawn iid $U[0,1]$.

\todo{Revisit below.}
Furthermore, we have
\begin{align}
\int \sum_{h \in \H} M_h(\theta)R_h(\theta) d\eta(\theta) & = \int \sum_{h \in \H} R_h(\theta)A_h(p_h^\theta) \prod_{h' \succ^\theta h}(1 - A_{h'}(p_{h'}^\theta)) d\eta^{IID}(\theta) \nonumber \\
&  = \eta^{IID}(\Theta) \sum_{k = 1}^\ell k q(1-q)^{k-1} \nonumber\\
& = \eta^{IID}(\Theta) (1 - (1-q)^\ell) AR(q). \label{eq:ar-denom}
\end{align}
Jointly, \eqref{eq:ar-num} and \eqref{eq:ar-denom} imply that
\begin{equation} \averagerank(M) = AR(q).\label{eq:avg-rank-ar} \end{equation}

Note that \eqref{eq:fill-to-capacity} implies that 
\begin{equation} \enrollment(\lambda,C) \leq C \text{ for all } \lambda \in \mathbb{R}_+,\label{eq:enr-ub} \end{equation} 
and therefore
\begin{equation} \eta^{IID}(\Theta) (1 - (1 - q)^\ell)  = |\H| \enrollment(I(0),C) \leq |\H| C. \end{equation} 
Define $\alpha$ and $\alpha'$ as the solutions to 
\begin{equation} \eta^{IID}(\Theta)(1 - (1-\alpha)^\ell) = |\H|C= \eta^{IID}(\Theta)(1 - e^{-\alpha' \ell}). \end{equation}
Then it follows that 
\begin{equation} q \leq \alpha \leq \alpha', \label{eq:q-alpha} \end{equation}
with the last inequality holding because $e^{-\alpha' \ell} \geq (1 - \alpha')^\ell$. Because $AR$ is decreasing by Lemma \ref{lem:ar}, equations \eqref{eq:avg-rank-ar} and \eqref{eq:q-alpha} imply that
\begin{equation} \averagerank(M) = AR(q) \geq AR(\alpha) = 1/\alpha -  \ell(\rho/C - 1) \geq 1/\alpha' +  \ell(1-\rho/C ), \end{equation} 
where the second equality follows from the definition of $\alpha$. Finally, noting that $\alpha' = \frac{-\log(1- C/\rho)}{\ell}$ completes the proof.

We now turn to the case where priorities are identical across schools. We define 
\begin{equation} q(u) = \vpois(\Lambda(u,C),C).\label{eq:q(u)} \end{equation} 
We will prove that the following chain of inequalities hold:
\begin{align}
\averagerank(M) & = \frac{1}{\enrollment(I(0),C)} \int_0^{\enrollment(I(0),C)} AR(q(u)) du \label{eq:step1} \\
& \leq \frac{1}{C} \int_0^C AR(q(u)) du \label{eq:step2} \\
& \leq \int_0^1 AR(q) dq \label{eq:step3} \\ 
& \leq 1 + \log(\ell). \label{eq:step4}
\end{align}

To evaluate $\averagerank(M)$, we note that
\begin{align}
\int \sum_{h \in \H} M_h(\theta)R_h(\theta) d\eta^{RSD}(\theta) & = |\H| \enrollment(I(0),C).\label{eq:messy-one} \\
\int \sum_{h \in \H} M_h(\theta)R_h(\theta) d\eta^{RSD}(\theta) & = \eta^{RSD}(\Theta) \int_{0}^1 (1-(1-\vpois(I(p),C))^\ell)  \nonumber AR(\vpois(I(p),C)) dp. 
\end{align}
We apply $u$-substitution to the latter integral, with $u = \enrollment(\I(p),C)$, so that %\todo{point to appropriate equations in school choice paper}
\[ \frac{du}{dp} = \vpois(\I(p),C) \I'(p) = - \frac{\eta^{RSD}(\Theta) }{|\H|}(1 - (1-\vpois(\I(p),C))^\ell).\]
This yields{\small
\[\eta^{RSD}(\Theta) \int_{0}^1 (1-(1-\vpois(I(p),C))^\ell) AR(\vpois(I(p),C)) dp = |\H| \int_{0}^{\enrollment(I(0),C)}AR(q(u)) du, \]}
which when combined with \eqref{eq:messy-one} yields \eqref{eq:step1}.

We now establish \eqref{eq:step2}. The function $q$ given in \eqref{eq:q(u)} is decreasing, as is $AR$ by Lemma \ref{lem:ar}.  Therefore, \eqref{eq:enr-ub} and Lemma \ref{lem:monotonicity}  imply \eqref{eq:step2}.

%In this case, a student's type $\theta$ is fully captured by their priority score, so we abuse notation and let $\theta \in [0,1]$ indicate this score. Because a student of type $\theta$ is admitted to each school with probability $q^\theta = \vpois(\I(\theta))$, the variable $R^{\theta}$ follows a (truncated) geometric distribution with success probability $q^\theta$. 
%
%Define 
%\[ \mu(\alpha) = \frac{1 - (1 - \alpha)^\ell}{\alpha},\]
%and note that $\E[\min(R^\theta,\ell)] = \mu(q^\theta)$ and $\p(R^\theta \leq \ell) = q^\theta \mu(q^\theta)$. It follows that
We move on to establishing \eqref{eq:step3}. Define $f: [0,1] \rightarrow \mathbb{R}_+$ implicitly by 
\begin{equation}\vpois(f(q),C) = q. \label{eq:f-mess} \end{equation}
We note that \todo{reference equation from school choice paper} \begin{equation} - \frac{d}{d\lambda}\vpois(\lambda,C) = \vpois(\lambda,C) - \vpois(\lambda,C-1) \leq \vpois(\lambda,C). \label{eq:vprime} \end{equation}
Therefore, by \eqref{eq:q(u)} and \eqref{eq:f-mess}, we have
\begin{align} \frac{1}{C} \int_0^C AR(q(u))du &  =\frac{1}{C} \int_0^1 AR(q) \frac{\vpois(f(q),C)}{\vpois(f(q),C) - \vpois(f(q),C-1)} dq.
\label{eq:ar-bound} \\
%& \geq \frac{1}{C} \int_0^1 AR(q) \frac{q}{\vpois(f(q),C)} dq
\end{align}

The function\footnote{\na{This looks like an inverse hazard rate. For large $C$, basically get normal cdf over normal density.}}
\[h(\lambda) = \frac{1}{C} \frac{\vpois(\lambda,C)}{\vpois(\lambda,C) - \vpois(\lambda,C-1) }\]
 is decreasing in $\lambda$ \todo{by what?}, from which it follows that $h(f(q))$ is increasing in $q$. Meanwhile, $AR$ is decreasing in $q$ by Lemma \ref{lem:ar}. It follows that 
\begin{equation} \int_0^1 AR(q) h(f(q)) dq \leq \int_0^1 AR(q) dq \int_0^1h(f(q)) dq  =  \int_0^1 AR(q) dq.\label{eq:splitting} \end{equation}
The final inequality follows because by \eqref{eq:f-mess} and \eqref{eq:fill-to-capacity} we have
\[ \int_0^1h(f(q)) dq = \frac{1}{C}\int_0^\infty \vpois(\lambda) d\lambda = 1.\]

Finally, we show \eqref{eq:step4}. We note that by $u$-substitution with $1-u = (1-q)^\ell$,
\[ \int_\varepsilon^1 \frac{\ell (1 - q)^\ell}{1 - (1-q)^\ell} dq = \int_{1-(1-\varepsilon)^\ell}^1 \frac{(1-u)^{1/\ell}}{u} du.\]
Thus, we can write 
\begin{align*} \int_{\varepsilon}^1 AR(q) dq = \int_\varepsilon^1 \frac{1}{q} - \frac{\ell (1 - 1)^\ell}{1 - (1-q)^\ell} dq & = \int_\varepsilon^{1 - (1-\varepsilon)^\ell} \frac{1}{q} d q + \int_{1 - (1-\varepsilon)^\ell}^1\frac{1-(1-u)^{1/\ell}}{u} du \\
& \leq \log\left( \frac{1 - (1-\varepsilon)^\ell}{\varepsilon} \right) + 1,
\end{align*}
where the second line follows by evaluating the first integral and bounding the second using the fact that $(1-(1-u)^{1/\ell})/u \leq (1 - (1-u))/u = 1$. Combining this with the fact that\footnote{Despite appearances, $AR$ is well-behaved at zero: for $q > 0$, 
\[ 1\leq \frac{1}{q} - \frac{\ell (1 - q)^\ell}{1 - (1-q)^\ell} \leq \frac{\ell+1}{2}.\]}
\[ \int_0^1 AR(q) dq = \lim_{\varepsilon \rightarrow 0} \int_{\varepsilon}^1 AR(q) dq. \]%\label{eq:arlim} \end{align}
implies \eqref{eq:step4}.\todo{\footnote{When $C = 1$, the analysis is fairly tight. I conjecture that tighter analysis would yield the nearly perfect bound $\averagerank(\I) \leq \log(e+\ell-1)$ Letting $\Delta = \int_0^1 \frac{1-(1-u)^{1/\ell}}{u} du$, it would suffice to show that $\Delta +\log(\ell) \leq \log(\ell + e-1)$, or equivalently $\Delta \leq \log(1 + (e-1)/\ell)$. \\
When $C > 1$, we can get a better bound, but not sure what it should be. As $C \rightarrow \infty$ with fixed $\rho/C$, should have average rank approaching $1$. One bound that appears to hold (but is guess, not analytical) is $AvgRank \leq 1+\log(\ell)/\sqrt{C}$.  Nice to note that $AR$ depends on $\mu$ but not $C$, and $h$ depends on $C$ but not $\mu$. Might be helpful to note that $\int_0^1 \ell(1-q)^\ell h(q) dq \leq 1$. As $C$ grows, second term converges to $\frac{1}{\sqrt{C}} \frac{q}{dnorm(qnorm(q))}$, but this latter term isn't integrable, so can't just pass limit through integral.}}

\end{proof}

\end{document}